\documentclass[11pt]{article}
\usepackage[a4paper, margin=1in]{geometry}

\usepackage{algorithm}
\usepackage{amssymb, amsmath, amsthm}
\usepackage{booktabs}
\usepackage{enumitem} 
\usepackage{float} 
\usepackage{graphicx}
\usepackage[colorlinks=true, linkcolor=LinkBlue, citecolor=LinkBlue, backref=page]{hyperref}
\usepackage{lmodern}
\usepackage{mathrsfs} 
\usepackage{multirow}
\usepackage{natbib} 
\usepackage{scalerel}
\usepackage{setspace}
\usepackage{soul} 
\usepackage{subfig}
\usepackage[T1]{fontenc}
\usepackage{verbatim}
\usepackage[dvipsnames]{xcolor}

\setlist{noitemsep, topsep=0pt} 
\definecolor{LinkBlue}{rgb}{.15, .25, .85} 
\providecommand{\lemmaname}{Lemma}
\providecommand{\propositionname}{Proposition}
\providecommand{\theoremname}{Theorem}
\providecommand{\corname}{Corollary}

\theoremstyle{plain} 

\theoremstyle{definition} 
\theoremstyle{definition} 

\newtheorem{lemma}{\protect\lemmaname}[section]

\providecommand{\algorithmname}{Algorithm}
\newfloat{algorithm}{tbp}{loa}
\floatname{algorithm}{\protect\algorithmname}

\newcommand{\eq}[1]{\begin{align*}#1\end{align*}} 

\newcommand{\R}{\ensuremath{\mathbb{R}}}

\newcommand{\E}{\operatorname{\mathbb{E}}}

\renewcommand{\P}{\operatorname{\mathbb{P}}}

\newcommand{\Bern}{\operatorname{Bern}}

\newcommand{\N}{\operatorname{N}}

\newcommand{\Unif}{\operatorname{Unif}}

\newcommand{\calB}{\mathcal{B}}

\newcommand{\bigO}{\mathcal{O}}

\newcommand{\iidsim}{\stackrel{iid}{\sim}}
\newcommand{\diff}{\, \mathrm{d} } 

\newcommand{\qaq}{\quad \text{and} \quad}

\newcommand{\TV}{{\sf TV}}
\newcommand{\symdiff}{ \operatorname{{\scalebox{.9}{$\bigtriangleup$}}}}

\newcommand{\tx}{\tilde{x}}
\newcommand{\ty}{\tilde{y}}

\newcommand{\n}{_{-1}}
\newcommand{\xo}{\xi_1}
\newcommand{\xn}{\xi_{-1}}

\newcommand{\zo}{\{0,1\}}
\newcommand{\cd}{\cdot}
\newcommand{\trx}{b_x x' + (1 - b_x) x}
\newcommand{\try}{b_y y' + (1 - b_y) y}
\newcommand{\g}{\, | \,}
\newcommand{\xy}{(x,y)}
\newcommand{\XY}{(X,Y)}
\newcommand{\xyp}{(x',y')}

\newcommand{\bxy}{(b_x,b_y)}

\newcommand{\bp}{\bar P}
\newcommand{\bq}{\bar Q}

\newcommand{\bb}{\bar B}

\newcommand{\Gmax}{\Gamma^\mathrm{max}}
\newcommand{\sas}{\ \text{and} \ }
\newcommand{\sm}{\setminus}

\begin{document}



\title{Couplings of the Random-Walk Metropolis algorithm}

\author{
	John O'Leary\thanks{
		    Department of Statistics, Harvard University, Cambridge, MA, USA.
		    Email: joleary@g.harvard.edu}
}
\maketitle

\begin{abstract}
	Couplings play a central role in contemporary Markov chain Monte Carlo methods and in the analysis
	of their convergence to stationarity. In most cases, a coupling must induce relatively fast meeting
	between chains to ensure good performance. In this paper we fix attention on the random walk
	Metropolis algorithm and examine a range of coupling design choices. We introduce proposal and
	acceptance step couplings based on geometric, optimal transport, and maximality considerations. We
	consider the theoretical properties of these choices and examine their implication for the meeting
	time of the chains. We conclude by extracting a few general principles and hypotheses on the design
	of effective couplings.
\end{abstract}

\section{Introduction} \label{sec:intro}

In commemorating the 50th anniversary of the Metropolis--Hastings (MH) algorithm,
\citet{dunson2020hastings} point to the unbiased estimation method of \citet{Jacob2020} as a leading
strategy for the parallelization of Markov chain Monte Carlo (MCMC) algorithms. However, they note a challenge:
while it is usually easy to find a transition kernel coupling with properties needed for this
approach, that choice is rarely unique, and the wrong selection can result in low estimator
efficiency. The design of efficient couplings is, as they write, ``an exciting direction that we
expect will see growing attention among practitioners.'' In this study we take up this important
question.

From the early days of Markov chain theory \citep[e.g.][]{doeblin1938expose, harris1955chains,
	pitman1976coupling, aldous1983random, Rosenthal1995}, couplings have played a key role in the
analysis of convergence to stationarity. In recent years they have also been used to formulate MCMC
diagnostics \citep{johnson1996studying, johnson1998coupling, biswas2019estimating}, variance
reduction methods, \citep{neal2001improving,Goodman2009,piponi2020hamiltonian}, and new sampling and
estimation strategies \citep{propp:wilson:1996, fill1997interruptible,
	neal1999circularly,flegal2012exact, glynn2014exact, Jacob2020, heng2019}. Couplings that produce
smaller meeting times generally yield better results in the form of tighter bounds, more variance
reduction, greater computational efficiency, or more precise estimators.

Thus, the design of efficient couplings has been an important question for almost the entire history
of the coupling method. When a coupling is not required to be co-adapted to the
chains in question, simple arguments show that a maximal coupling of the chains exists
and results in meeting at the fastest rate allowed by the coupling inequality
\citep{griffeath1975maximal, goldstein1979maximal}. However when the coupling must be implementable and
Markovian, maximal couplings are known only in special cases
\citep{burdzy2000efficient, hsu2013maximal, bottcher2017markovian}. Markovian couplings are easy to
work with and are required for many of the applications above, but they are rarely maximal.

In this study we consider transition kernel couplings of the Random Walk Metropolis (RWM) algorithm
\citep{Metropolis1953}, which is perhaps the oldest, simplest, and best-understood MCMC method.
Transition kernel couplings \citep[][chap. 19]{douc2018markov} are Markovian by construction.
Explicit and implementable couplings of the RWM kernel seem to originate with
\citet{johnson1998coupling}. These methods were taken up in \citet{Jacob2020}, which found that
apparently minor differences in coupling design can have significant implications for meeting times,
especially in relation to the dimension of the state space. In this paper we continue this line of
inquiry and take a pragmatic approach to the question of coupling design. We ask: what options are
available for coupling the RWM kernel, how do these choices affect meeting times, and what lessons
can we learn from this simple case?

We begin by introducing the essential ingredients of an RWM kernel coupling. First, we consider
proposal distribution couplings, devoting some attention to maximal couplings of the multivariate
normal distribution. Next, we turn to coupling at the accept/reject step. Any coupling of the RWM kernel can be realized as a
proposal coupling followed by an acceptance step coupling \citep{o2021couplings}, so this focus on separate proposal and
acceptance couplings involves no loss of generality. We conclude with a range of simulation
exercises to understand how various coupling design options affect meeting times. We conclude with some
stylized facts and advice on the construction of efficient couplings for the RWM algorithm and beyond.

\section{Setting and notation} \label{sec:defn}

Throughout the following we write $a \wedge b := \min(a,b)$, $a \vee b := \max(a,b)$, and $\mathcal{L}(Z)$ for
the law of a random variable $Z$. We write
$\Bern(\alpha)$ for the Bernoulli distribution on $\zo$ with $\P(\Bern(\alpha)=1) = \alpha$,
$\N(\mu, \Sigma)$ for the multivariate normal distribution with mean $\mu$ and covariance matrix
$\Sigma$, and $\N(z ; \mu, \Sigma)$ for the density of this distribution evaluated at a point $z$.

Fix a target distribution $\pi$ on a state space $(\R^d, \calB)$, where $\calB$ is the Borel
$\sigma$-algebra on~$\R^d$. Let $Q : \R^d \times \calB \to [0,1]$ be a proposal kernel. Thus
$Q(x,\cd)$ is a probability measure for all $z \in \R^d$ and $Q(\cd, A):\R^d \to [0,1]$ is
measurable for all $A \in \calB$. We interpret $Q(x,A)$ as the probability of proposing some point
$x' \in A$ when the current state is $x$. Assume $\pi$ has density $\pi(\cd)$ and $Q(x,\cd)$ has
density $q(x,\cd)$ for $x \in \R^d$, all with respect to Lebesgue measure. The MH acceptance ratio \citep{hastings:1970} is then defined as $a(x, x') := 1 \wedge \tfrac{q(x',x)}{q(x, x')} \tfrac{\pi(x')}{\pi(x)}$. In this study we focus on the RWM algorithm with
multivariate normal proposal increments, so $Q(x,\cd) = \N(x, I_d \, \sigma_d^2)$ for all $x
\in \R^d$. In this case ${ q(x, \tx) = q(\tx, x) }$ for all $x,x' \in \R^d$, and ${ a(x, \tx ) =
	1 \wedge (\pi(\tx ) / \pi(x)) }$.

We construct an MH chain $(X_t)$ as follows. First we initialize the chain with a draw $X_0$ from an
arbitrary distribution $\pi_0$ on $(\R^d, \calB)$. At each iteration $t$ we draw $\tx \sim Q(x,
\cdot)$, where $x = X_t$ is the current state of the chain. We then draw an acceptance indicator
$b_x \sim \Bern(a(x,x'))$ and set $X_{t+1} := \trx$. It is often convenient to realize the
acceptance indicator draw by taking $U \sim \Unif$ and $b_x := 1(U \leq a(x,x'))$. The chain $(X_t)$
defined above will have a transition kernel $P$ defined by $P(x,A) := \P(X_{t+1} \in A \g X_t = x)$ for
$x \in \R^d$ and $A \in \calB$.

For any probability measures $\mu$ and $\nu$ on $(\R^d, \calB)$, we say that a probability measure
$\gamma$ on $(\R^d \times \R^d, \calB \otimes \calB)$ is a coupling of $\mu$ and $\nu$ if $\gamma(A
\times \R^d) = \mu(A)$ and $\gamma(\R^d \times A) = \nu(A)$ for all $A \in \calB$. We write
$\Gamma(\mu,\nu)$ for the set of all such couplings of $\mu$ and $\nu$. Next suppose that $(X_t)$
and $(Y_t)$ are both Markov chains defined on the same probability space and that both evolve
according to the RWM transition kernel $P$ defined above. We say that $(X_t, Y_t)$ follows a
transition kernel coupling $\bp$ based on $P$ if there exists a joint kernel $\bar{P} : (\R^d \times
\R^d) \times (\calB \otimes \calB) \to [0,1]$ with $\bp(\xy, \cd) \in \Gamma(P(x,\cd),P(y,\cd))$ for
all $x,y \in \R^d$. We write $\Gamma(P,P)$ for the set of all such kernel couplings. The limitation
to couplings of $(X_t)$ and $(Y_t)$ that can be expressed in the form above is not a trivial one, as
described further in \citet{Kumar2001}.

We write $\tau = \min(t : X_t = Y_t)$ for the first time the chains meet. Couplings $\bp$ with the
property that $\P(\tau < \infty) = 1$ are called successful. To obtain successful couplings, we
generally need a proposal kernel coupling with $\P(x' = y' \g x,y) > 0$ from at least some state
pairs $(X_t,Y_t) = (x,y)$. This will lead us to consider maximal couplings of the proposal
distributions, which achieve the highest possible probability $\P(x'=y' \g x,y)$ for each $x,y \in \R^d$.
Couplings $\bar{P}$ with the property that $X_t = Y_t$ for all $t \geq \tau$ are called sticky and
are also our subject of interest here. \citet{rosenthal1997faithful} and \citet{dey2017note} point
out that stickiness is a non-trivial property, even for Markovian couplings. However, the couplings we consider can always be made sticky by requiring $x'=y' \sim Q(x,\cd) = Q(y,\cd)$ and $V=U \sim \Unif$ if $x=y$.

In this paper we consider a range of coupling options for the RWM transition kernel. Our goal is
to understand the implications of these options for the distribution of $\tau$ and especially for
the value of its mean $\E[\tau]$. The average meeting time serves as a convenient summary of the
meeting rate and plays a specific role in the efficiency of the estimators described in
\citet{Jacob2020}. We will focus on transition kernel couplings that arise by separately coupling
the proposals $(x',y')$ and the uniform draws $(U,V)$ underlying the acceptance indicators
$\bxy = (1(U \leq a(x,x')), 1(V \leq a(y,y')))$. This strategy is fully general except with respect
to the acceptance indicator coupling, as noted in \citet{o2021couplings}.

When it is unlikely to cause confusion, we write $(x,y) = (X_t, Y_t)$ for the current state of a
pair of coupled MH chains, $(x', y') = (x + \xi, y + \eta)$ for the proposals, and $(X, Y) =
(X_{t+1},Y_{t+1})$ for the next state pair. We write $r = ||x-y||$ for the Euclidean distance
between $x$ and $y$ and $m = (x+y)/2$ for their midpoint. We write $e = (y- x)/r$ for the unit
vector pointing from $x$ to $y$, an important direction in many of the constructions described below.
Finally, for any $z \in \R^d$ we write $z_1 = e'z \in \R$ for the $e$ component of $z$ and $z\n =
(I_d - ee') z \in \R^d$ for the projection of $z$ onto the subspace orthogonal to $e$. Thus we can
express any vector $z \in \R^d$ as $z = e z_1 + z\n$. See Figure~\ref{fig:cpl:geom} for an
illustration of these quantities.

\begin{figure}[t]
	\centering
	\includegraphics[trim={1cm .5cm 2.75cm 2.5cm}, clip=true,
	width=0.55\linewidth]{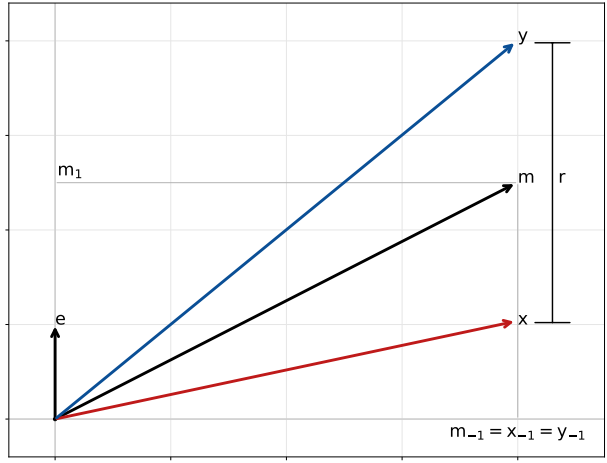}
	\caption{Coupled chain notation and geometry.
		We denote current points by $x, y \in \R^d$, their separation by $r \geq 0$, and their midpoint by $m \in \R^d$.
		The unit vector $e$ points in the direction from $x$ to $y$, $m_1 \in \R$ gives the $e$ component of $m$, and $m\n = x\n = y\n
		\in \R^d$ gives the component of these three vectors which is orthogonal to $e$.
		\label{fig:cpl:geom}}
\end{figure}

\section{Maximal coupling foundations}

To obtain finite meeting times we generally need $\P(x'=y' \g x,y) > 0$ from at least
some state pairs $\xy$ with $x \neq y$. One solution is to draw $(\tx , \ty )$ from a maximal
coupling $\bq \in \Gamma(Q,Q)$. A coupling of ${\tx \sim Q(x,\cdot)}$ and ${\ty \sim Q(y, \cdot)}$
is said to be maximal if it achieves the upper bound given by the coupling inequality, $\P(\tx = \ty \g x,y) \leq
1-||Q(x, \cdot) - Q(y, \cdot)||_\mathrm{TV} = 1 - \sup_{A \in \calB} |Q(x,A) - Q(y,A)|$. See
\citet[chap. 1.4]{thorisson2000coupling} or \citet[chap 4.]{Levin2017} for discussion of this bound
and its applications. For any probability distributions $\mu$ and $\nu$ on $(\R^d, \calB)$,
we write $\Gmax(\mu, \nu)$ for the set of all maximal couplings of $\mu$ and $\nu$.
Maximal couplings of the proposal distribution make an appealing starting
point, but note that their use is neither necessary nor sufficient to maximize $\P(X=Y \g x,y)$.
\citet{Lee2020} also observe that the
variance of the computational cost to draw from a maximal coupling can blow up when $r = \lVert y-x
\rVert \to 0$. In such cases one may prefer to use a slightly non-maximal coupling over a maximal
one.

The following result, closely related to \citet{douc2018markov}, Theorem 19.1.6 and Proposition D.2.8,
shows that maximality comes with significant constraints on a coupling's behavior:

\begin{lemma}
	\label{lem:maxpdf}
	Let $\bq(\xy,\cd)$ be a maximal coupling of $Q(x,\cd)$ and $Q(y,\cd)$,
	distributions with densities $q(x,\cd)$ and $q(y,\cd)$ on $\R^d$.
	If $\xyp \sim \bq(\xy,\cd)$, then for all $A \in \calB$,
	\eq{
		\P( x' \in A, x' = y' \g x,y) = \P( y' & \in A, x' = y' \g x,y ) = \int_A q(x, z) \wedge q(y, z) \diff z \\
		\P( x' \in A, x' \neq y' \g x,y ) &= \int_A 0 \vee ( q(x, z) - q(y, z) ) \diff z \\
		\P( y' \in A, x' \neq y' \g x,y) &= \int_A 0 \vee ( q(y, z) - q(x, z) ) \diff z.
	}
\end{lemma}

We can obtain $\P(x' = y' \g x,y)$ by evaluating the first equation at $A = \R^d$. This meeting
probability takes a particularly simple form for multivariate normal distributions, as we see in the
following extension of \citet[chap. 3.3]{Pollard2005}:

\begin{lemma}
	\label{lem:meetprob}
	If $\xyp$ follows any maximal coupling of
	$\N(x, I_d \sigma^2_d)$ and $\N(y, I_d \sigma^2_d)$, then
	\eq{
		\P(x' = y' \g x,y) = \P\Big(\chi^2_1 \geq \tfrac{||y-x||^2}{4 \sigma^2_d} \Big).
	}
\end{lemma}

\begin{proof}
	Recall that we write $\N(z ; \mu, \Sigma)$ for the density of $\N(\mu, \Sigma)$ and have
	defined $m = (y+x)/2$, $r = \lVert y - x \rVert$, $e = (y-x)/r, z_1 = e' z$ and $z\n = (I_d -
	ee')z$ for all $z \in \R^d$.
	In general $ \N(z; \mu, I_d \sigma^2) = \N(z_1 ; \mu_1, \sigma_d^2)
	\N(z\n; \mu\n, I_{d-1} \sigma_d^2)$.
	We can also decompose $x$ and $y$ into $e$ and $e^\perp$ parts
	according to $x = m - e \tfrac{r}{2} = (m_1 - \tfrac{r}{2}) e + m\n$ and $y = m + e \tfrac{r}{2} =
	(m_1 + \tfrac{r}{2}) e + m\n$. Combining these expressions with Lemma~\ref{lem:maxpdf} yields
	the desired conclusion:
	\eq{
		\P(\tx = \ty \g x,y) & = \int \N(z; x, I_d \sigma^2_d) \wedge \N(z; y, I_d \sigma_d^2) \diff z \\ &
		= \iint \left( \N(z_1; -\tfrac{r}{2}, \sigma_d^2) \wedge \N(z_1; \tfrac{r}{2}, \sigma_d^2) \right)
		\N(z\n; m\n, I_{d-1} \sigma_d^2)  \diff z_1 \diff z\n \\ & = \int_{-\infty}^\infty \N(z_1;
		-\tfrac{r}{2}, \sigma_d^2) \wedge \N(z_1; \tfrac{r}{2}, \sigma_d^2) \diff z_1 = 2 \int_{0}^\infty
		\N(z_1; -\tfrac{r}{2}, \sigma_d^2) \diff z_1 \\ & = 2 \P( \N(0,1) \geq \tfrac{r}{2 \sigma_d}) =
		\P(\chi^2_1 \geq \tfrac{r^2}{4 \sigma_d^2}).
	}\\[-5.6em]
\end{proof}

An important implication of Lemma \ref{lem:meetprob} is that as we increase the dimension $d$, the
separation $r = ||y-x||$ needed to hold $\P(x'=y' \g x,y)$ constant must vary in proportion to
$\sigma^2_d$. Under the typical RWM assumption that $\sigma^2_d = \ell^2/d$, this means $r$ must
shrink at a rate $1/\sqrt{d}$ to maintain a constant probability of meeting proposals. This inverse
square-root condition plays a crucial role in determining the dimension scaling behavior of
different couplings as we will observe in the simulations of Section~\ref{sec:rwmsims}. Note that the meeting probability derived in Lemma \ref{lem:meetprob} also
admits the following useful inequalities:

\begin{figure}
	\centering
	\includegraphics[trim={1cm .5cm 2.75cm 2.5cm}, clip=true,
	width=0.5\linewidth]{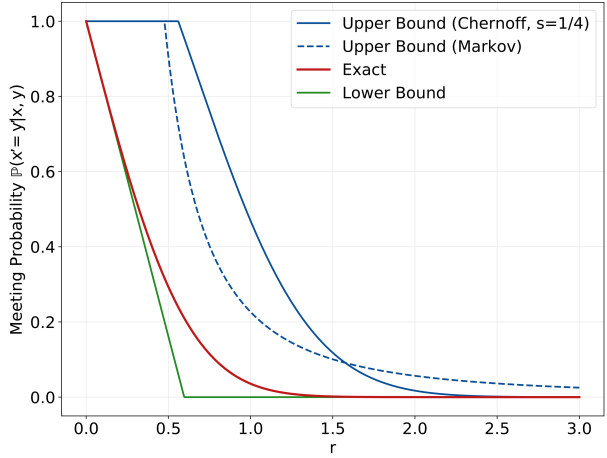}
	\caption{
		Meeting probability and bounds for maximal couplings of normal distributions on $\R$.
		The red line (`Exact') gives $\P( x'=y' \g x,y) $ when $\xyp$ follows a maximal coupling of $\N(x,1)$ and $\N(y,1)$, based on the results on Lemma~\ref{lem:meetprob}. The blue and green lines give the bounds derived in Lemma~\ref{lem:mpbounds}. Here $r = \lVert y - x \rVert$. The exact meeting probabilities involve the complementary CDF of the $\chi^2_1$ distribution, so it can be analytically more convenient to use the given bounds.
		\label{fig:cpl:propmeet}}
\end{figure}

\begin{lemma} \label{lem:mpbounds} Under any maximal coupling of $x' \sim \N(x, I_d
	\sigma^2_d)$ and $y' \sim \N(y, I_d \sigma^2_d)$, we have
	\eq{
		1 - \sqrt{\tfrac{2}{\pi}} \tfrac{||y-x||}{2 \sigma_d}
		\leq \P(x' = y' \g x,y)
		\leq \tfrac{ 4 \sigma^2_d}{ \lVert y- x \rVert^2}
		\qaq
		\P( x'=y' \g x,y)
		\leq \tfrac{1}{\sqrt{1-2s}} \exp \Big( -s \tfrac{||y-x||^2}{4 \sigma^2_d} \Big)
	}
	for $s \in (0,1/2)$.
\end{lemma}

\begin{proof} For the lower bound, let $\phi(z) = \N(z ; 0,1)$ be the standard
	normal density and let $\Phi(z)$ be the corresponding cumulative distribution function.
	Since $\Phi(0) = 1/2$ and $\phi(z) \leq
	1/\sqrt{2\pi}$ for all $z$, then for $a>0$ we may write $\Phi(a) = \int_{-\infty}^{a} \phi(z) \diff
	z \leq \tfrac{1}{2} + \tfrac{a}{\sqrt{2 \pi}}$. This expression rearranges to $1 -
	\sqrt{\tfrac{2}{\pi}} a \leq 2 (1 - \Phi(a))$, and plugging in $a = r / (2 \sigma_d)$ yields
	the desired lower bound.
	The first upper bound follows directly from Markov's inequality: \eq{ \P(x' = y' \g
		x,y) = \P(\chi_1^2 \geq \tfrac{r^2}{4 \sigma^2_d}) \leq \frac{\E[\chi_1^2]}{r^2 / (4 \sigma^2_d)} =
		\frac{4 \sigma^2_d}{r^2}. } The second upper bound is due to Chernoff's inequality, $\P(\chi^2_1
	\geq a) \leq e^{-sa} \E[e^{s \chi^2_1}]$ for all $s > 0$. We have $\E[e^{s \chi^2}] =
	1/\sqrt{1-2s}$ for $s < 1/2$, so plugging in $a = r^2 / (4 \sigma^2_d)$ yields the desired
	expression.
\end{proof}

In Figure \ref{fig:cpl:propmeet} we plot the value of $\P(\tx = \ty \g x,y)$ as derived in
Lemma \ref{lem:meetprob} along with the upper and lower bounds from Lemma \ref{lem:mpbounds}. We
observe that while the lower bound is tight at $r=0$ and in the limit as $r \to \infty$, the upper
bounds only become tight in the large-$r$ limit. When needed, sharper upper and lower bounds can be obtained from more precise Gaussian tail inequalities, see e.g.  \citet[chap. 7]{Abramowitz1988} and
\citet{Duembgen2010}.

We close by noting that if we can produce draws from one maximal coupling, we can often
transform these into draws from a maximal coupling of a related pair of distributions.
Recall that for any measure $\mu$ on $(\R^d,
\calB)$ and measurable function $f: \R^d \to \R^d$, the pushforward measure $f_\star
\mu$ on $(\R^d, \calB)$ is defined by $f_\star \mu (A) := \mu(f^{-1}(A))$ for all $A \in \calB$.
Also, if $x' \sim \mu$ then $f(x') \sim f_\star \mu$. Thus we have the following:

\begin{lemma}
	\label{lem:transf}
	Suppose $\mu$ and $\nu$ are probability measures on $(\R^d, \calB)$ and
	let $f: \R^d \to \R^d$ be a homeomorphism. $\xyp$ follows a maximal coupling of $\mu$ and $\nu$
	if and only if $(f(x'), f(y'))$ follows a maximal coupling of $f_\star \mu$ and $f_\star \nu$.
\end{lemma}

\begin{proof}
	First, $(x',y') \in \Gamma(\mu, \nu)$ if and only if $(f(x'), f(y')) \in \Gamma(f_\star \mu, f_\star \nu)$ since $f$ is a bijection.
	Also
	\eq{
		\lVert f_\star \mu - f_\star \nu \rVert_\TV
		= \sup_{A \in \calB} |\mu(f^{-1}(A)) - \nu(f^{-1}(A))|
		= \sup_{B \in \calB'} |\mu(B) - \nu(B)|
		= \lVert \mu - \nu \rVert_\TV.
	}
	Here $\calB' = \{ f^{-1}(A) : A \in \calB \}$, and $\calB' = \calB$ since $f$ is a homeomorphism.
	Since $f$ is a bijection, we also have  $\P(x' = y') = \P(f(x') = f(y'))$.
	Thus $(x',y')$ will achieve the coupling inequality bound exactly when $(f(x'),f(y'))$ does. Thus we have
	shown that the former pair follows a maximal coupling of $\mu$ and $\nu$ if and only if the latter follows a maximal coupling
	of $f_\star \mu$ and $f_\star \nu$.
\end{proof}

Lemma~\ref{lem:transf} allows us to efficiently draw from and analyze the maximal independent
coupling of distributions like $\N(x,\Sigma)$ and $\N(y,\Sigma)$ in terms of the maximal independent
coupling of $\N(0, I_d)$ and $\N(\Sigma^{-1/2}(y-x), I_d)$. It can also be useful in the design of
couplings when the proposal kernel arises from a deterministic but well-behaved function of a
multivariate normal random variable, as in the case of Hamiltonian Monte Carlo
\citep{Duane:1987,Neal1993,neal2011mcmc}.


\section{Proposal step couplings}
\label{sec:props}

In this section we describe a range of proposal kernel couplings $\bq$ based on the RWM proposal
kernel $Q(z,\cd) = \N(z, I_d \sigma^2_d)$ on $\R^d$. If $\xyp = (x + \xi, y + \eta) \sim \bq(\xy,
\cd)$, then marginally $\xi, \eta \sim \N(0, I_d \sigma^2_d)$. These increments can exhibit
a complex dependence pattern, and $(\xi,\eta)$ need not be multivariate normal. The simplest
option, however, is the independent coupling ${\xi, \eta \iidsim \N(0, I_d \sigma^2_d)}$. One step more
complex is the synchronous or `common random numbers' coupling ${\xi = \eta \sim \N(0, I_d
	\sigma^2_d)}$. As noted in \citet{givens1984class} and \citet{knott1984optimal}, the synchronous
coupling minimizes the expected squared distance $\E[\lVert x' - y' \rVert^2 ]$ among all joint
distributions with $x' \sim \N(x,I_d \sigma^2_d)$ and $y' \sim \N(y,I_d \sigma^2_d)$. We comment
further on optimal transport couplings below.

Another slightly more complex option is the simple reflection coupling, in which $\xi \sim \N(0, I_d
\sigma^2_d)$ and ${\eta = (I_d - 2 e e') \xi}$. With the notation $z_1 = e' z$ and $z\n = (I_d -
ee')z$ for any $z \in \R^d$, we note that the reflection coupling yields $\eta_1 =-\xo$ and $\eta\n
= \xn$. Thus $\eta$ is the reflection of $\xi$ over the hyperplane $\mathcal{H} = \{z : ||z - x|| =
||z-y||\} = \{z : z_1 = m_1\}$. Taking this geometric logic a step further, we can also consider
the full-reflection coupling in which $\xi \sim \N(0, I_d \sigma^2_d)$ and $\eta = - \xi$. This
coupling maximizes $\E[\lVert y' - x' \rVert^2 ]$ just as the synchronous coupling minimizes it. The
independent, synchronous, reflection, and full-reflection couplings are easy to draw from and
straightforward to analyze. They also differ dramatically in the covariance and transport properties
that they establish between $x'$ and $y'$, their interactions with various accept/reject procedures,
and thus the degree of contraction they produce between
coupled chains. However each of these couplings has the property that if $x \neq y$ then $x' \neq y'$
almost surely. This implies $X \neq Y$, so exclusive reliance on these couplings cannot yield
$\P(\tau < \infty) = 1$ unless $X_0 = Y_0$.

\subsection{The maximal independent coupling}

Suppose $\xyp \sim \bq(\xy,\cd)$ for some $\bq \in \Gmax(Q,Q)$. One consequence of
Lemma~\ref{lem:maxpdf} is that all maximal couplings exhibit the same distribution of $x'$ and
$y'$ given $x'=y'$. In particular, each of these variables will have conditional density $q^m_{xy}(z):= q(x,z)
\wedge q(y,z) / \int q(x,w) \wedge q(y,w) \diff w$. We refer to the
distributions of $x'$ and $y'$ given $x' \neq y'$ as the residuals of $\bq(\xy, \cd)$. In light of
the above, we differentiate between various maximal couplings according to the behavior of these
residuals, i.e. according to the distribution of $\xyp$ conditional on $x' \neq y'$.

The first and perhaps most famous maximal coupling
was introduced by \citet{Vaserstein1969} and termed the
$\gamma$-coupling by \citet{lindvall2002lectures}.
It is the unique maximal coupling with the property that $x'$ and $y'$ are independent when $x' \neq y'$.
Thus we call this the maximal coupling with independent residuals, or simply the maximal independent coupling.
When $Q(z,\cd) = \N(z, I_d \sigma^2_d)$ for $z \in \R^d$, Lemmas~\ref{lem:maxpdf} and \ref{lem:meetprob}
imply that this coupling approximates the independent coupling of $Q(x,\cd)$ and $Q(y,\cd)$
as $r=||y-x|| \to \infty$.

\begin{algorithm}
	\caption{Draw from the maximal independent coupling of $Q(x,\cd)$ and $Q(y,\cd)$ \label{alg:maxbasic}}
	\begin{enumerate}
		\item Draw $x' \sim Q(x,\cd)$ and $W_x \sim \Unif$
		\item If $W_x \, q(x,x') \leq q(y,x')$, set $y' = x'$
		\item Else:
		\begin{enumerate}
			\item Draw $\tilde y \sim Q(y,\cd)$ and $W_y \sim \Unif$
			\item If $W_y \, q(y, \tilde y) > q(x, \tilde y)$, set $y' = \tilde y$
			\item Else go to 3(a)
		\end{enumerate}
		\item Return $(x',y')$
	\end{enumerate}
\end{algorithm}

The references above prove that one can draw from the maximal independent coupling by using
the rejection sampling procedure described in Algorithm \ref{alg:maxbasic}. This method is simple and
versatile, although it suffers from a loss of efficiency as a function of dimension. In our setting,
each normal density evaluation requires $\bigO(d)$ computations, and these costs can be a factor in
algorithmic performance in high dimensions or when the number of iterations required to obtain a
valid $y' $ draw is large. Algorithm \ref{alg:maxquick} offers an alternative, which exploits the symmetries
and factorization properties of the multivariate normal distribution. It provides a more
efficient way to draw from the maximal coupling of these distributions, and it also lends itself to
extensions and variations as we consider below. Lemma \ref{lem:maxquick} establishes the validity of this
algorithm.

\begin{algorithm}
	\caption{Draw from the maximal independent coupling of $\N(x, I_d \sigma_d^2)$ and $\N(y, I_d \sigma_d^2)$. \label{alg:maxquick}}
	\begin{enumerate}
		\item Compute $e = (y-x)/\lVert y - x \rVert$, $m = (y+x)/2$, $x_1 = x'e$, and $y_1 = y'e$
		\item Draw $(x'_1, y'_1)$ from the maximal independent coupling of $\N(x_1, \sigma_d^2)$ and $\N(y_1, \sigma_d^2)$
		using Algorithm~\ref{alg:maxbasic}
		\item Independently draw $\tilde x \sim \N(x, I_d \sigma^2_d)$ and $\tilde y \sim \N(y, I_d \sigma^2_d)$
		\item Set $x'\n = (I_d - ee') \tilde x \sas y'\n = (I_d - e e') \tilde y$
		\item Set $x' = \tx_1 e + \tx\n$. If $\tx_1=\ty_1$ set $y'=\ty_1 e + \tx\n$, else set $y'= \ty_1 e + \ty\n$
		\item Return $(x',y')$
	\end{enumerate}
\end{algorithm}

\begin{lemma}
	\label{lem:maxquick}
	The output of Algorithm~\ref{alg:maxquick} is distributed according to the maximal independent coupling of $\N(x, I_d \sigma_d^2)$ and $\N(y, I_d \sigma_d^2)$.
\end{lemma}

\begin{proof}
	First we show that the output $\xyp$ of Algorithm~\ref{alg:maxquick} follows a coupling of
	$\N(x, I_d \sigma_d^2)$ and $\N(y, I_d \sigma_d^2)$. For $x'$ we have $x_1' \sim \N(e'x,\sigma^2_d)$
	and $x \n' \sim \N( (I_d - ee') x, (I_d - ee') \sigma^2_d)$ with independence between $x_1'$ and $x'\n$.
	Thus $x' = x'_1 e + x \n' \sim \N(x, I_d \sigma^2_d)$. For $y'$, note that
	$y'\n \sim \N(y\n, (I_d - e e') \sigma^2_d)$ whether or not $x'_1=y'_1$. This is trivial when $x'_1\neq y'_1$.
	When $x'_1 = y'_1$ we have
	$\E[y'\n \g x,y, x_1'=y_1'] = (I_d - e e') x = (I_d - e e') (m - \tfrac{r}{2}e)
	= (I_d - e e') m = (I_d - e e') (m + \tfrac{r}{2} e) = y \n$.
	Also, $y_1' \sim \N(y'e, \sigma^2_d)$ and $y'\n$ are independent, so we conclude $y' \sim \N(y, I_d \sigma^2_d)$.

	Next, we show that $\xyp$ follows a maximal coupling.
	By Lemma \ref{lem:meetprob}, draws from a maximal coupling of
	$\N(x, I_d \sigma^2_d)$ and $\N(y, I_d \sigma^2_d)$ must meet with probability
	$\P(\chi^2_1 \geq \tfrac{\lVert y - x \rVert^2}{ 4\sigma_d^2})$.
	By construction we have $x' = y'$ if and only if $x'_1 = y'_1$.
	Applying Lemma \ref{lem:meetprob} to the maximal coupling of
	$\N(x_1,\sigma_d^2)$ and $\N(y_1,\sigma_d^2)$ shows that meeting occurs with probability
	$\P(\chi^2_1 \geq \tfrac{ (y_1 - x_1)^2}{ 4\sigma_d^2} )$.
	We also have $y_1 - x_1 = e'(y-x) = (y - x)'(y-x)/\lVert y - x \rVert = \lVert y - x \rVert$, so meeting occurs at the maximal rate.

	Finally, we observe that $x'$ and $y'$ are independent conditional on $x' \neq y'$. This holds for $x_1$ and $y_1$ since these are drawn from a maximal independent coupling on $\R$, and it holds for $x\n$ and $y\n$ since in the relevant case these are defined using independent random variables.
	Thus Algorithm~\ref{alg:maxquick} produces draws from the maximal independent coupling of
	$\N(x,I_d \sigma^2_d)$ and $\N(y,I_d \sigma^2_d)$.
\end{proof}

Overall, the meeting time associated with a transition kernel coupling depends on that coupling's probability of producing a meeting at each step
together with the dynamics of the chains conditional on not meeting.
It is often a good idea to control the variance of $y' - x'$ when $x' \neq y'$, to reduce the tendency of the chains to push apart when meeting does not occur.
This motivates what we call the maximal coupling with semi-independent residuals, or the maximal semi-independent coupling, which we define in Algorithm~\ref{alg:maxsemi}. This algorithm differs from the maximal independent coupling in that it has $x'\n = y'\n$ whether or not $x'_1 = y'_1$. The validity of this algorithm follows from essentially the same argument as that of Lemma~\ref{lem:maxquick}.

\begin{algorithm}
	\caption{Draw from the maximal semi-independent coupling of $\N(x, I_d \sigma_d^2)$ and $\N(y, I_d \sigma_d^2)$. \label{alg:maxsemi}}
	\begin{enumerate}
		\item Compute $e = (y-x)/\lVert y - x \rVert$, $m = (y+x)/2$, $x_1 = x'e$, and $y_1 = y'e$
		\item Draw $(x'_1, y'_1)$ from the maximal independent coupling of $\N(x_1, \sigma_d^2)$ and $\N(y_1, \sigma_d^2)$
		using Algorithm~\ref{alg:maxbasic}
		\item Draw $\tilde z \sim \N(m, I_d \sigma^2_d)$, set $z' \n = (I_d - ee') \tilde z$, $x' = x'_1 e + z'\n$, and $y' = y'_1 e + z'\n$
		\item Return $(x',y')$
	\end{enumerate}
\end{algorithm}

\subsection{Optimal transport couplings}

As noted above, the joint distribution of $(\tx, \ty)$ given $\tx \neq \ty$ plays an important role
in determining the distribution of meeting times. This is especially
important since Lemma \ref{lem:meetprob} shows that the probability $\P(x'=y' \g x,y)$ of meeting
proposals must be small until the chains are relatively close. Thus it is natural to consider not
just ways to limit the variance of $y'-x'$ when $x' \neq y'$, but methods for making this quantity
as small as possible.

Given a metric $\delta$ on $\R^d$, we say that $\bq(\xy, \cd) \in \Gamma(Q(x,\cd),Q(y,\cd))$ is an
optimal transport coupling if $\bq(\xy, \cd)$ minimizes $\E_{\xyp \sim \tilde Q}[\delta(x',y;)]$
among all couplings $\tilde Q \in \Gamma(Q(x,\cd),Q(y,\cd))$. In this study we set $\delta(x',y') =
\lVert y' - x' \rVert^2$. Below, we show how to construct an optimal transport coupling between the
residuals of a maximal coupling of $\N(x,I_d \sigma^2_d)$ and $\N(y,I_d \sigma^2_d)$. Optimal
transport couplings are not usually available in closed form, but the symmetries of the multivariate
normal distribution present an opportunity. We begin with the following result in one dimension:

\begin{lemma}
	\label{lem:ot1d} Suppose $\bq(\xy,\cd) \in \Gmax(\N(x,\sigma^2),\N(y,\sigma^2))$, and define the
	residual distributions $\mu(A) := \P(x' \in A \g x' \neq y', x,y)$ and $\nu(A) = \P(y' \in A \g x'
	\neq y', x,y)$ where $\xyp \sim \bq(\xy,\cd)$ and $A \in \calB$. Let $\Phi_x$ and $\Phi_y$ be the
	cumulative distribution functions of $\mu$ and $\nu$ on $\R$, and define the transport map $t_{xy}(x') :=
	\Phi_y^{-1}(\Phi_x(x'))$. If $x' \sim \mu$, then $(x',t_{xy}(x'))$  is an optimal transport coupling
	of $\mu$ and $\nu$. Also $\Phi_x$ and $\Phi_y$ have the
	functional forms given in the proof below.
\end{lemma}

\begin{proof}
	The main result is due to the cumulative distribution function characterization of optimal transport
	maps for non-atomic distributions on $\R$, see e.g. \citet[chap. 3.1]{rachev1998mass}. If $x' \sim
	\mu$ and $y' \sim \nu$ then by Lemma~\ref{lem:maxpdf}, $x'$ and $y'$ the following CDFs:
	\eq{
		\Phi_x(x') = \begin{cases}
			\tfrac{F_x(x' \wedge m) - F_y(x' \wedge m) }{F_x(m) - F_y(m)} & \text{if } x < y \\
			1 - \tfrac{ F_x(x' \vee m) - F_y(x' \vee m) }{F_x(m) - F_y(m)} & \text{if } x \geq y
		\end{cases}
		\qquad
		\Phi_y(y') = \begin{cases}
			\tfrac{F_y(y' \wedge m) - F_x(y' \wedge m) }{F_y(m) - F_x(m)} & \text{if } y < x \\
			1 - \tfrac{ F_y(y' \vee m) - F_x(y' \vee m) }{F_y(m) - F_x(m)} & \text{if } y \geq x.
		\end{cases}
	}
	Here $m = (x+y)/2$ and $F_z(\cd)$ is the CDF of $\N(z, \sigma^2)$ for $z \in \R$.
\end{proof}

We say that $\bq$ is a maximal coupling with optimal transport residuals, or a maximal optimal transport coupling, if $\bq(\xy,\cd) \in \Gmax(Q(x,\cd),Q(y,\cd))$ and if the residuals of $\bq(\xy,\cd)$ follow an optimal transport coupling.
The result above suggests an algorithm for drawing from the maximal optimal transport coupling of one-dimensional normal distributions. See Algorithm~\ref{alg:ot1d} for the details of this method and Lemma~\ref{lem:ot1dvalidity} for a proof of its validity.

\begin{algorithm}
	\caption{Draw from the maximal optimal transport coupling of $\N(x, \sigma^2)$ and $\N(y, \sigma^2)$. \label{alg:ot1d}}
	\begin{enumerate}
		\item Draw $x' \sim \N(x, \sigma^2)$ and $W_x \sim \Unif$
		\item If $W_x \sim N(x' ; x,\sigma^2) \leq \N(x' ; y, \sigma^2)$, set $y' = x'$
		\item Else set $y' = t_{xy}(x')$ using the transport map $t_{xy}$ as defined in Lemma~\ref{lem:ot1d}
		\item Return $(x',y')$
	\end{enumerate}
\end{algorithm}

\begin{lemma}
	\label{lem:ot1dvalidity}
	The output of Algorithm~\ref{alg:ot1d} follows
	a maximal optimal transport coupling of $\N(x, \sigma^2)$ and $\N(y, \sigma^2)$.
\end{lemma}

\begin{proof}
	$x' \sim \N(x, \sigma^2)$ by construction.
	By Lemma~\ref{lem:maxpdf} and the validity of Algorithm~\ref{alg:maxbasic},
	we have $\P(y' \in A, y' = x' ) = \int_A \N(y' ; y, \sigma^2) \wedge \N(y' ; x, \sigma^2) \diff y'$
	for $A \in \calB$.
	Lemmas~\ref{lem:maxpdf} and \ref{lem:ot1d} also imply
	$\P(y' \in A, y' \neq x' ) = \int_A (\N(y' ; y, \sigma^2) - \N(y' ; x, \sigma^2)) \vee 0 \diff y'$ for $A \in \calB$.
	Together these imply $y' \sim \N(y, \sigma^2)$, so $\xyp$ follows some coupling of $\N(x,\sigma^2)$ and $\N(y,\sigma^2)$. Finally, we note that Algorithm~\ref{alg:ot1d} has exactly the same probability of
	$x'=y'$ as Algorithm~\ref{alg:maxbasic} does when $Q(z, \cd) = \N(z, \sigma^2)$. We know that the coupling implemented in Algorithm~\ref{alg:maxbasic} is maximal, so we conclude that the present one is as well.
\end{proof}

Finally, we combine the result of Lemma~\ref{lem:ot1dvalidity} with the logic of
Algorithm~\ref{alg:maxsemi} to obtain an algorithm for drawing from the maximal coupling with
optimal transport residuals on $\R^d$. See Algorithm~\ref{alg:mot} for a statement of this method
and Lemma~\ref{lem:mot} for a proof of its validity.

\begin{algorithm}
	\caption{Draw from the maximal optimal transport coupling of $\N(x, I_d \sigma_d^2)$ and $\N(y, I_d \sigma_d^2)$ \label{alg:mot}}
	\begin{enumerate}
		\item Compute $e = (y-x)/\lVert y - x \rVert$, $m = (y+x)/2$, $x_1 = x'e$, and $y_1 = y'e$
		\item Draw $(x'_1, y'_1)$ from the maximal optimal transport coupling of $\N(x_1, \sigma_d^2)$ and $\N(y_1, \sigma_d^2)$ using Algorithm~\ref{alg:ot1d}
		\item Draw $\tilde z \sim \N(m, I_d \sigma^2_d)$, set $z' \n = (I_d - ee') \tilde z$,
		$x' = x'_1 e + z'\n$, and $y' = y'_1 e + z'\n$
		\item Return $(x',y')$
	\end{enumerate}
\end{algorithm}

\begin{lemma}
	\label{lem:mot}
	The output of Algorithm~\ref{alg:mot} follows a maximal optimal transport coupling of $\N(x, I_d \sigma^2_d)$ and $\N(y, I_d \sigma^2_d)$.
\end{lemma}

\begin{proof}
	The proof that $(x',y')$ follows a maximal coupling of $\N(x,I_d \sigma^2_d)$ and $\N(y, I_d \sigma^2_d)$ is almost identical to the argument of Lemma~\ref{lem:maxquick}, except we now use the same draw for $y'\n = x'\n$ rather than independent draws $x'\n, y'\n \sim \N(m\n, (I_d - e e') \sigma^2_d)$.
	To see that $(x', y')$ follows an optimal transport coupling conditional on $x' \neq y'$,
	we apply Theorem 2.1 of \citet{knott1984optimal}.
	That result says that if we can write $y' = T_{xy}(x')$ such that $y'$ has the correct distribution and
	$\partial T_{xy}(x') / \partial x'$ is symmetric and positive definite, then $(x', T_{xy}(x'))$ is an optimal transport
	coupling for $\delta(x',y') = \lVert y' - x' \rVert^2$. In this case we have $y' = T_{xy}(x') = t_{xy}(x_1') e + x'\n$. Symmetry follows immediately and positive definiteness follows since $t_{xy}$ is monotonically increasing.
\end{proof}

\subsection{The maximal reflection coupling}

Another coupling in the spirit of the previous section is the maximal coupling with reflection
residuals, also called the maximal reflection coupling. It is the maximal analogue to the
reflection coupling defined near the beginning of Section~\ref{sec:props}, and it has previously been
considered in \citet{eberle2019quantitative, bou2018coupling}, and
\citet{Jacob2020}. We say that $\xyp \sim \bq(\xy, \cd)$ is a maximal reflection coupling of
$\N(x,I_d\sigma^2_d)$ and $\N(y,I_d\sigma^2_d)$ if it is a maximal coupling and if
$x' \neq y'$ implies $\eta = (I_d - 2 e e') \xi$,
where we define $\xi = x' - x$ and $\eta = y' - y$.
When $x'=y'$, the maximal reflection coupling yields the same distribution of $\xyp$ as any other maximal
coupling, and when $x' \neq y'$ it reflects the increments of each chain over the hyperplane
equidistant between $x$ and $y$.

This coupling is related to the reflection coupling of diffusions described in
\citet{lindvall1986coupling, Eberle2011, hsu2013maximal}, and other studies.
These continuous-time reflection couplings
are sometimes maximal couplings of processes,
in the strong sense that they produce the fastest meeting times allowed by the coupling inequality.
We will see that using the maximal reflection coupling for RWM proposals also delivers good
meeting time performance. In our setting this seems to arise from a felicitous interaction
between reflection couplings and the Metropolis accept/reject step. Understanding the analogy between the continuous- and discrete-time settings remains an interesting open question, especially for reflection couplings.

\begin{algorithm}
	\caption{Draw from the maximal reflection coupling of $\N(x, \sigma^2)$ and $\N(y, \sigma^2)$ \label{alg:ref1d}}
	\begin{enumerate}
		\item Draw $x' \sim \N(x, \sigma^2)$ and $W_x \sim \Unif$
		\item If $W_x \sim N(x' ; x,\sigma^2) \leq \N(x' ; y, \sigma^2)$, set $y' = x'$
		\item Else set $\xi = x' - x$, $\eta = - \xi$, and $y' = y + \eta$
		\item Return $(x', y')$
	\end{enumerate}
\end{algorithm}

As with the maximal independent and optimal transport couplings, we describe an efficient method for
drawing from the maximal reflection coupling.
We begin with Algorithm~\ref{alg:ref1d}, which yields draws from the maximal reflection coupling on $\R$. The validity of this algorithm is established in
\citet[sec. 2]{bou2018coupling} and
\citet[sec. 4]{Jacob2020}. Algorithm~\ref{alg:refl} produces draws from the general form of this coupling on $\R^d$, and we establish the validity of this algorithm in Lemma~\ref{lem:refl}.
For the algorithm and its validity proof, recall that we have defined $z_1 := e'z$ and $z\n = (I_d - e e') z$ for any $z \in \R^d$.

\begin{algorithm}[b]
	\caption{Draw from the maximal reflection coupling of $\N(x, I_d \sigma_d^2)$ and $\N(y, I_d \sigma_d^2)$ \label{alg:refl}}
	\begin{enumerate}
		\item Compute $e = (y-x)/\lVert y - x \rVert$, $m = (y+x)/2$, $x_1 = x'e$, $y_1 = y'e$, and $m\n = (I_d - e e') m$
		\item Draw $(x'_1, y'_1)$ from the maximal reflection coupling of $\N(x_1, \sigma_d^2)$ and $\N(y_1, \sigma_d^2)$,
		by the method of Algorithm~\ref{alg:ref1d}
		\item Draw $\zeta \sim \N(0, I_d \sigma^2_d)$ and set $x' = x'_1 e + m\n + \zeta\n$, and $y' = y'_1 e + m\n + \zeta\n$
		\item Return $(x',y')$
	\end{enumerate}
\end{algorithm}

\begin{lemma}
	\label{lem:refl}
	The output $\xyp$ of Algorithm~\ref{alg:refl} is distributed according to a maximal reflection coupling of $\N(x, I_d \sigma_d^2)$ and $\N(y, I_d \sigma_d^2)$.
\end{lemma}

\begin{proof}
	Essentially the same argument as in Lemmas~\ref{lem:maxquick} and \ref{lem:mot} establishes that $(x',y')$
	follows a maximal coupling. For the reflection condition
	we recall that $y = m + r/2 e$ and $x = m - r/2 e$,
	which implies $y = y_1 e + m \n$ and $x = x_1 e + m\n$.
	Thus $y' - y = (y_1' - y_1) e + \zeta\n$ and $x' - x = (x_1' - x_1) e + \zeta\n$.
	We also have $(I_d - 2 e e') e = - e$ and $(I_d - 2 e e')(I_d - e e') = (I_d - e e')$.
	By the definition of Algorithm~\ref{alg:ref1d}, $y_1' - y_1 = - (x_1' - x_1)$ when $x_1' \neq y_1'$.
	Thus $x' \neq y'$ implies
	\eq{
		(I_d - 2 e e')(x'-x) = - (I_d - 2 e e') ( y_1' - y_1) e + (I_d - 2 e e') \zeta\n = (y_1
		- y_1) e + \zeta \n = y' - y.
	}
	We conclude that $\xyp$ satisfies the reflection condition when $x' \neq y'$,
	and so the output of Algorithm~\ref{alg:refl} follows a maximal reflection coupling of $\N(x,I_d \sigma^2_d)$ and $\N(y, I_d \sigma^2_d)$.
\end{proof}

\subsection{Hybrid couplings}

It is also possible to choose among the coupling strategies described above -- or any valid coupling
of the proposal distributions -- depending on the current state pair $\xy$. As implied by
Lemma~\ref{lem:meetprob}, maximal couplings of Gaussian distributions have very little chance of
producing $x' = y'$ unless $r = \lVert y - x \rVert$ is relatively small. Thus we can deploy one
coupling method such as a maximal coupling when $r$ is below some threshold and a different coupling
when $r$ is above it. This is reminiscent of the two-coupling strategies used in
\citet{smith2014gibbs, pillai2017kac}, and \citet{bou2018coupling}. This approach can be deployed to
produce faster meeting between chains. It can also simplify some theoretical arguments since, for example,
the simple reflection coupling is simpler to analyze than its maximal coupling counterpart.

\section{Acceptance step couplings}
\label{sec:accrej}

Recall that we write $P$ for the MH transition kernel generated by the proposal kernel $Q$ and
acceptance rate function $a$. We construct our MH kernel coupling $\bp \in \Gamma(P,P)$ as follows.
First, we draw $(X_0, Y_0)$ such that $X_0, Y_0 \sim \pi_0$ for some arbitrary initial distribution $\pi_0$.
We begin each iteration $t$ by drawing proposals $\xyp \sim \bq(\xy,\cd)$, where
$(x,y) = (X_t, Y_t)$. Then we draw acceptance indicators $\bxy$ from
some joint distribution $\bb(\xy, \xyp)$ on $\zo^2$. Finally we set $(X_{t+1}, Y_{t+1}) = (X,Y)$
where $X = \trx$ and $Y = \try$. For any $x,y \in \R^d$ and $A \in \calB \otimes \calB$, we write
$\bp(\xy, A) := \P(\XY \in A \g x,y)$ for the resulting joint transition distribution. We want
$\bp(\xy,\cd) \in \Gamma(P(x,\cd),P(y,\cd))$, and which implies a few constraints on the acceptance
indicator coupling $\bb(\xy,\xyp)$.

Given any mapping $\bb$ from current state and proposal pairs $\xy, \xyp$ to probabilities
on $\zo^2$, we can define joint acceptance rate functions $a_x(\xy, \xyp) := \P(b_x =
1 \g x,y,x',y')$ and $a_y(\xy, \xyp) := \P(b_y = 1 \g x,y,x',y')$, where $\bxy \sim \bb(\xy,\xyp)$.
These definitions make $\bb(\xy,\xyp)$ a coupling of $\Bern(a_x(\xy,\xyp))$ and
$\Bern(a_y(\xy,\xyp))$. We want the transition pair $\XY$ defined above to imply $X \sim P(x,\cd)$
and $Y \sim P(y,\cd)$ conditional on $\xy$. In \citet{o2021couplings}, this is shown to hold if $\xyp \sim
\bq(\xy, \cd)$ for some $\bq \in \Gamma(Q,Q)$ and if
\eq{
	\P(b_x = 1 \g x,y, x') & = \E[a_x(\xy,\xyp) \g x,y,x'] = a(x,x') \quad \text{for $Q(x,\cd)$-almost all $x'$} \\
	\P(b_y = 1 \g x,y, y') & = \E[a_y(\xy,\xyp) \g x,y,y'] = a(y,y')\quad \text{for $Q(y,\cd)$-almost all $y'$.}
}
These conditions are intuitive, but they allow for relatively complicated forms of $a_x$, $a_y$, and
$\bb$. For example, this flexibility is used in \citet{OLeary2020} to formulate an acceptance indicator coupling
$\bb$ which yields a maximal transition kernel coupling $\bp \in \Gamma(P,P)$ any time it is used
with a maximal proposal coupling $\bq \in \Gamma(Q,Q)$. For now, we focus on acceptance indicator
couplings with $a_x(\xy,\xyp) = a(x,x')$ and $a_y(\xy,\xyp) = a(y,y')$.
The resulting acceptance couplings take a simple form, as described in the following result.

\begin{lemma}
	\label{lem:accind}
	Suppose $\bxy \sim \bb(\xy,\xyp) \in \Gamma(\Bern(a(x,x')), \Bern(a(y,y')))$
	for state pairs $\xy,\xyp \in \R^d \times \R^d$.
	Then for some $\rho_{xy} \in [0 \vee (a(x,x') + a(y,y') -1), a(x,x') \wedge a(y,y')]$,
	\begin{alignat*}{2}
		& \P(b_x=1, b_y=1)  = \rho_{xy}  \quad
		&& \P(b_x=1, b_y=0) = a(x,x') - \rho_{xy} \\
		& \P(b_x = 0, b_y = 1)  = a(y,y') - \rho_{xy}  \quad
		&& \P(b_x = 0, b_y = 0) = 1 - a(x,x') - a(y,y') + \rho_{xy}.
	\end{alignat*}
\end{lemma}

\begin{proof}
	Set $\rho_{xy} := \P(b_x=1, b_y=1)$. The values for $\P(b_x=1,b_y=0)$ and $\P(b_x=0,b_y=1)$ follow from the margin conditions,
	and then the value of $\P(b_x=0,b_y=0)$ follows from the requirement that $\sum_{i,j \in \zo} \P(b_x=i,b_y=j) = 1$. The constraints on $\rho_{xy}$ follow from the requirement that all of these joint probabilities must fall in $[0,1]$.
\end{proof}

Note for example that independent draws
$b_x \sim \Bern(a(x,x'))$ and $b_y \sim \Bern(a(y,y'))$ imply
$\rho_{xy} = a(x,x')a(y,y')$. This satisfies the given
bounds, since $\rho_{xy} = a(x,x')a(y,y') \leq a(x,x') \wedge a(y,y')$ and $\rho_{xy} = a(x,x') a(y,y') \geq a(x,x')
+ a(y,y') - 1$ from the fact that $(1-a(x,x'))(1-a(y,y')) \geq 0$.

The two chains can only meet if $\tx = \ty$ is proposed and both proposals are accepted. This
suggests that we should maximize the probability of $\bxy = (1,1)$
by choosing $\rho_{xy} = a(x,x') \wedge a(y,y')$. By Lemma~\ref{lem:accind}, this also maximizes
the probability of $\bxy = (0,0)$ and of $b_x = b_y$.
Thus, this $\rho_{xy}$ corresponds to using the maximal coupling of $\Bern(a(x,x'))$ and $\Bern(a(y,y'))$, which is unique in this case.
Simulation results suggest that some couplings tend to produce contraction between chains when both
proposals are accepted, no change in the separation between chains when both are rejected, and an increase
in separation when one chain is accepted and the other is rejected. This further argues for drawing $\bxy$ from
its maximal coupling conditional on $\xy$ and $\xyp$.

Write $A \symdiff B = (A \sm B) \cup (B \sm A)$ for the symmetric difference of $A,B \in \calB$.
As described in Section~\ref{sec:intro}, it is convenient to describe acceptance indicators and their couplings in terms of uniform random variables. In particular we have the following:

\begin{lemma}
	\label{lem:unifcpl}
	Fix $a_x, a_y \in [0,1]$.
	$\tilde B \in \Gamma(\Bern(a_x),\Bern(a_y))$
	if and only if there exists a coupling $\bar U \in \Gamma(\Unif, \Unif)$ such that
	$\bxy \sim \tilde B$ for $b_x = 1(U \leq a_x)$ and $b_y = 1(V \leq a_y)$.
	In particular, $\P(b_x=b_y \g x,y,x',y')$ is maximized when $U = V$
	and minimized when $V = 1-U$.
\end{lemma}

\begin{proof}
	Suppose $\bar U \in \Gamma(\Unif,\Unif)$ and $\bxy$ are defined as in the statement above.
	$b_x \sim \Bern(a_x)$ since $\P(b_x = 1) = \P(U \leq a_x) = a_x$, and similarly for $b_y$.
	Thus the law of $\bxy$ is a coupling of $\Bern(a_x)$ and $\Bern(a_y)$.
	For the converse, by Lemma~\ref{lem:accind} any coupling $\tilde B$ will be characterized by $\rho = \P(b_x=b_y=1)$.
	Thus we must find a coupling $\bar U \in \Gamma(\Unif,\Unif)$ such that
	if $(U,V) \sim \bar U$ then $\P(U \leq a_x, V \leq a_y) = \rho$.
	One such coupling is the distribution on $[0,1]^2$ with density
	\eq{
		f(u,v) = \left\{\hspace{-1ex}
		\begin{array}{llll}
			\frac{\rho}{a_x a_y}                                      &\text{if } u \leq a_x, v \leq a_y  \hspace{3em}
			& \frac{a_x- \rho}{a_x (1-a_y)}                     &\text{if } u \leq a_x, v > a_y     \\
			\frac{a_y- \rho}{(1-a_x) a_y}                         &\text{if } u > a_x, v \leq a_y     \hspace{3em}
			& \frac{1 - a_x - a_y + \rho}{(1-a_x)(1- a_y)} &\text{if } u > a_x, v > a_y.
		\end{array}
		\right.
	}
	Note that when $U = V \sim \Unif$, we obtain $\P(b_x=b_y=1) = a(x,x') \wedge a(y,y')$, the maximal value of $\rho$,
	and when $1 - V = U \sim \Unif$ we achieve $\P(b_x=b_y=1) = 0 \vee (a(x,x') + a(y,y') - 1)$, the minimal value of $\rho$.
	By Lemma~\ref{lem:accind}, the probability of $b_x=b_y$ is maximized when $\rho$ is maximized and minimized when $\rho$ is minimized.
\end{proof}

While the acceptance indicator couplings described above are appealing in their simplicity, we may
wonder if we can do better by adapting our choice of $\rho_{xy}$ depending on the current state pair
$\xy$ or the proposals $\xyp$. In particular, we consider an `optimal transport' approach to
selecting $\rho_{xy}$, in which we aim to minimize the expected distance between state $X=b_x x' +
(1-b_x) x$ and $Y=b_y y' + (1-b_y) y$ after the joint accept/reject step. For each $x,y,x',y' \in
\R^d$, we solve \eq{ \min_{\rho} \{\ \E[\delta(X,Y) \g x,y,x',y']\ : \ \rho \in [(a(x,x') + a(y,y')
	- 1 )\vee 0, a(x,x') \wedge a(y,y')] \}. } As above, we set $\delta(X,Y) = \lVert Y -X \rVert^2$.
Note that this is a linear program with linear constraints in $\rho$, so for typical proposal
couplings $\bq$ the above will almost surely have solution either $\rho = 0 \vee (a(x,x') + a(y,y')
- 1)$ or $\rho = a(x,x') \wedge a(y,y')$. Qualitatively, the lower-bound solution
will be optimal when $\delta(x,y')$ and $\delta(x',y)$ are small relative to $\delta(x,y)$ and
$\delta(x',y')$. Below, we see that this is uncommon for the proposal distributions we consider.


\section{Simulations}
\label{sec:rwmsims}

We now consider a set of simulations on the relationship between coupling design and meeting times,
with a focus on the role of dimension. High-dimensional target distributions are a common
challenge in applications of MCMC, and previous studies such as \citet{Jacob2020} suggest that
apparently similar couplings can produce unexpected and sometimes dramatic differences in meeting
behavior with increasing dimension $d$. We expect that theory will someday provide definitive
guidance on the design of couplings that scale well with dimension. For now, simulations like the
following suggest the use of some couplings over others and offer a range of hypotheses for further analysis.

Our simulations target the standard multivariate normal distribution $\pi_d = \N(0, I_d)$ with a
range of dimensions $d$. We focus on the RWM algorithm with proposal kernel $Q(x, \cd) = \N(x, I_d
\sigma^2_d)$ with $\sigma^2_d = \ell^2/d$ and $\ell = 2.38$. This form of proposal variance yields
an acceptance rate converging to the familiar 0.234 value as $d \to \infty$, and diffusion limit
arguments suggest that this choice produces rapid or even optimal mixing behavior at stationarity
\citep{Gelman1996, roberts1997weak} and in the transient phase \citep{Christensen2005,
	Jourdain2014}. Meeting times typically depend on both the coupling and the marginal behavior of each
chain. Our multivariate normal setting is a particularly simple one, but it allows us to isolate the
effect of each coupling decision without concern about the mixing behavior of the marginal kernel.

\begin{table}
	\caption{Average meeting times (1000 replications each, $d=10$) \label{tab:combos}}
	\centering
	\begin{tabular}{lcccccc}
		\toprule
		\hfill Acceptance Coupling & \multicolumn{2}{c}{$V = U$} & \multicolumn{2}{c}{Independent} & \multicolumn{2}{c}{$V = 1 - U$} \\
		Proposal Coupling & Avg $\tau$ &  S.E. & Avg $\tau$ &  S.E.& Avg $\tau$ &  S.E. \\
		\midrule
		Maximal Reflection        &    30 &  0.8 &          51 &  1.4 &        68 &   2.0 \\
		Maximal Semi-Independent  &    54 &  1.5 &          85 &  2.4 &       105 &   3.3 \\
		Maximal Optimal Transport &   104 &  3.0 &         155 &  4.6 &       183 &   5.7 \\
		Maximal Independent       &   279 &  8.5 &         302 &  9.4 &       354 &  11.2 \\
		\bottomrule
	\end{tabular}
\end{table}

We begin by considering a full grid of proposal and acceptance coupling combinations in $d=10$. We
initialize each chain using an independent draw from the target distribution. We then iterate until
meeting occurs and record the observed meeting time over 1000 replications with each pair of coupling
options. The averages and standard deviations of these meeting times $\tau$ appear in
Table~\ref{tab:combos}. We will consider each of these options in more detail below, but for now it
is important to note a few facts. First, even in low dimension like $d=10$, we already observe
more than an order of magnitude difference in the meeting times associated with the best and worst
coupling combinations. In the simulations below we generally evaluate high-performance
couplings over a much wider range of dimensions than low-performance couplings to avoid unnecessary
computational expense.

Second, among the acceptance couplings considered above, the $U=V$ coupling consistently outperforms
the $U,V \iidsim \Unif$ coupling, which consistently outperforms the $V= 1 - U$ coupling. These
relationships hold for each choice of proposal coupling. A similar relationship exists among the
proposal couplings, with the maximal reflection coupling delivering the best meeting times and the
maximal independent coupling delivering the worst ones, again for any choice of acceptance indicator
coupling. These robust and monotone relations also hold in higher and lower dimensions and with
other initialization and proposal variance options. Although the meeting times shown here arise
from a complex interplay of proposal and acceptance behavior, these simulations suggest that some
options can be regarded as generally better or worse than others.

In this exercise and in many of the simulations below, we initialize chains using independent draws
from the target distribution. When $\pi_d = \N(0, I_d)$, the initialization method does not seem to
affect the relative performance of the couplings considered below. This may not hold for all
targets, e.g. mixture distributions with well-separated modes. The development of couplings and
initialization strategies to address especially challenging targets stands as an important topic for
future research.

\subsection{Proposal couplings}

We now consider the relationship between proposal couplings and meeting times in more detail. As
above, we initialize each chain with an independent draw from the target distribution. Here we use
the $U=V\sim \Unif$ coupling at the Metropolis step, which maximizes the probability of making the
same accept/reject decision for both chains. As illustrated in Table~\ref{tab:combos}, this
acceptance coupling seems to produce the fastest meeting times for a range of proposal couplings.
For each proposal coupling and dimension, we run 1000 pairs of chains until meeting occurs. The
average meeting times from this test appear in Figure~\ref{fig:ft:prop}.

In Figure \ref{fig:mt:prop:nonrefl}, we show the average meeting times for the maximal couplings
with independent, optimal transport, semi-independent, and reflection residuals, as defined in
Section~\ref{sec:props}. Figure \ref{fig:mt:prop:refl} presents the corresponding results for the
maximal-reflection coupling and for hybrid couplings that deploy the maximal reflection coupling
when $r = ||y-x|| < \bar{r}/\sqrt{d}$ but use the simple reflection coupling when the chains are
further apart. We consider hybrid couplings with a range of values of the cutoff parameter
$\bar{r}$.

\begin{figure}
	\centering
	\subfloat[\label{fig:mt:prop:nonrefl} Maximal proposal distribution couplings]{
		\includegraphics[trim={.8cm .5cm 2.75cm 2.5cm}, clip=true,
		width=0.45\textwidth]{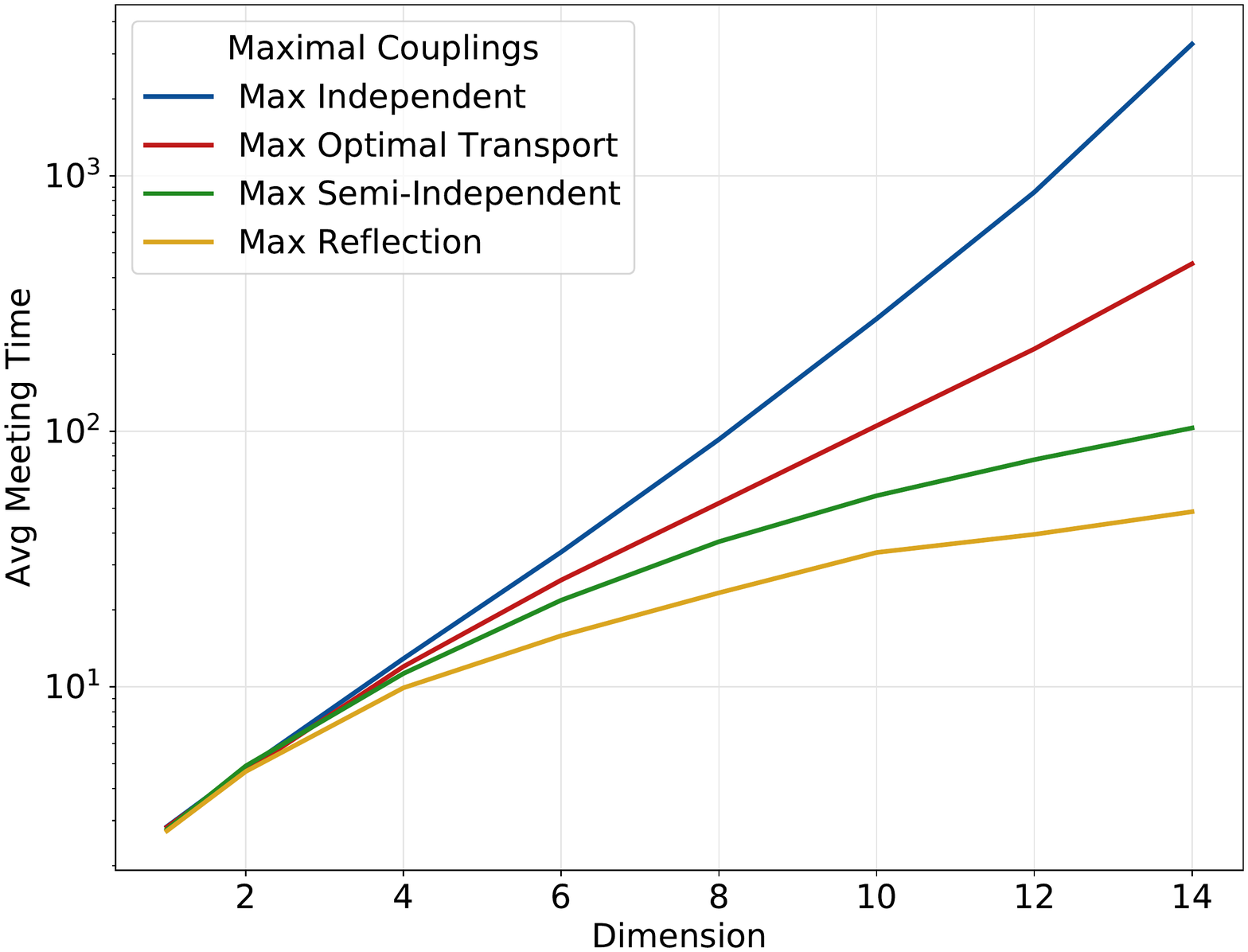} }
	\subfloat[\label{fig:mt:prop:refl} Maximal and hybrid reflection couplings]{
		\includegraphics[trim={.8cm .5cm 2.75cm 2.5cm}, clip=true,
		width=0.45\textwidth]{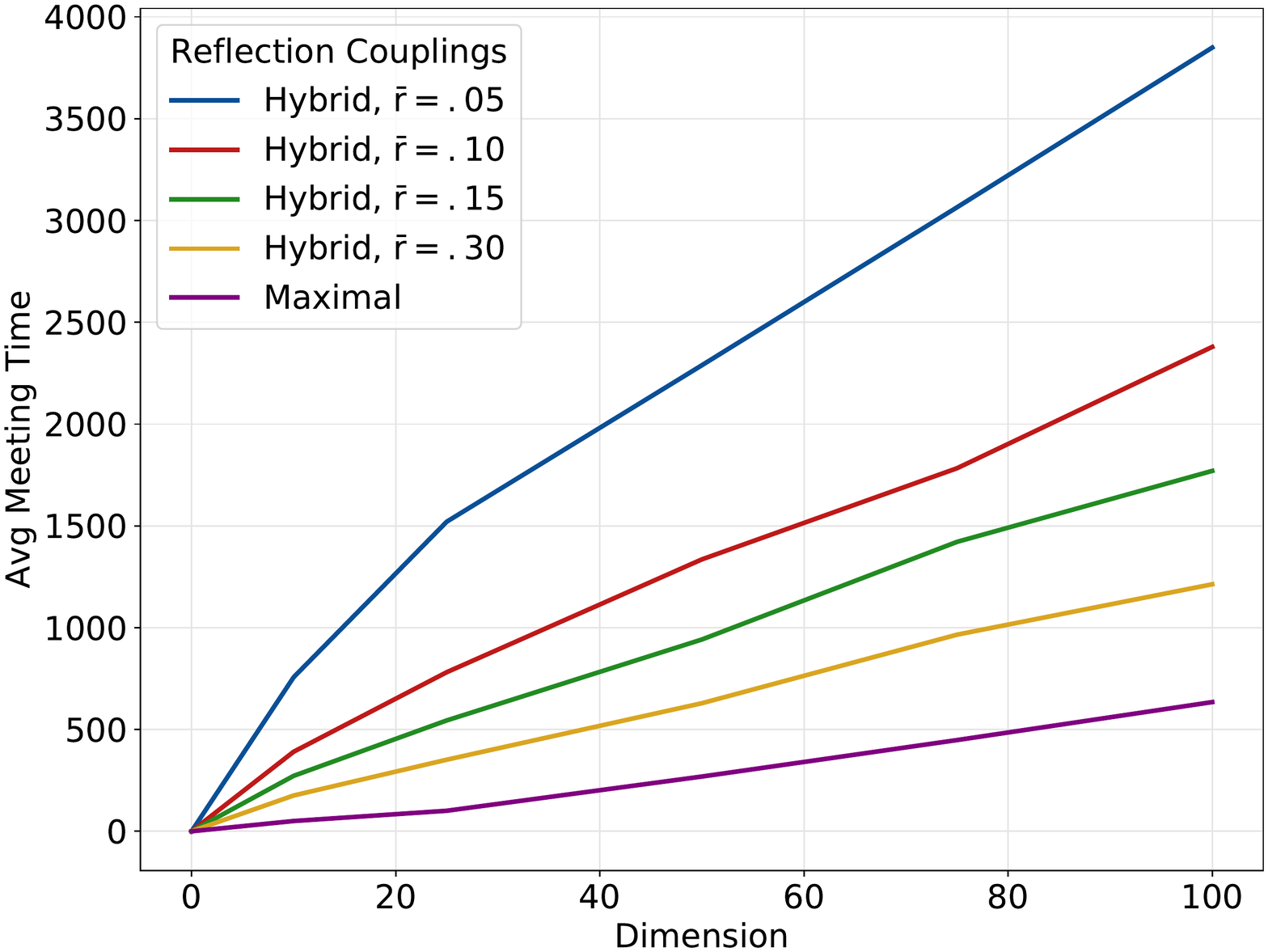} }
	\caption{\label{fig:ft:prop}~Scale behavior of the average meeting time of coupled RWM chains as a
		function of the dimension of the target distribution $\pi = \N(0, I_d)$, under proposal
		distribution coupling options defined in Section~\ref{sec:props}. The chains are initialized at
		independent draws from the target and a $V=U$ coupling is used at the Metropolis step. Left:
		Scaling under four maximal couplings of the proposal distribution. Right: Scaling under the maximal
		reflection coupling and simple/maximal reflection coupling hybrids under a range of cutoff parameters $\bar r$.}
\end{figure}

These results suggest that meeting times grow exponentially in dimension under the maximal coupling
with independent residuals, close to linearly in dimension under the maximal reflection coupling,
and somewhere in between for the other two maximal couplings. The hybrid couplings and maximal
reflection coupling show a similar order of dependence on dimension, and the hybrid couplings
display an inverse relationship between average meting time and $\bar{r}$. This reflects an
increasing number of missed opportunities to meet under the hybrid couplings, since smaller values
of $\bar r$ result in more situations when $r \geq \bar r / \sqrt d$ even though $r$ is small enough to produce a reasonable probability of meeting under a maximal coupling.

Any maximal coupling of proposal distributions produces meeting proposals with the same probability,
as a function of $R_t = \lVert Y_t - X_t \rVert$. Thus the variation in average $\tau$ reflects
differences in coupled chain dynamics conditional on not meeting. The degree of contraction between
chains seems to play a particularly important role. As noted above, just after Lemma~\ref{lem:meetprob}, the
chains must be within a distance $r_d = \bigO(1/\sqrt{d})$ to maintain a fixed probability of
proposed meetings as $d$ increases. Thus the combination of a proposal and acceptance coupling must
generate contraction $\E[R_{t+1} - R_t \g X_t, Y_t] < 0$ to within a range $r_d$ to avoid a fall-off
in meeting probability as a function of dimension. The results above suggest that some proposal
couplings do this better than others.

To visualize this behavior, we run 1000 pairs of coupled chains under a range of maximal and
non-maximal couplings, as described in  Section~\ref{sec:props}. We fix $d=100$, initialize chains
independently from the target, and use the $U=V \sim \Unif$ coupling at the accept reject/step. We
run all pairs of chains for 2500 iterations and use the sticky coupling described in
Section~\ref{sec:defn} to maintain $X_t = Y_t$ for $t \geq \tau$. Finally, we compute $R_t = \lVert
Y_t - X_t \rVert$ and plot the average distance over replications as a function of the iteration
$t$. See Figure~\ref{fig:mt:prop:rtrace} for these results.

\begin{figure}
	\centering
	\includegraphics[trim={.8cm .5cm 2.75cm 2.5cm}, clip=true,
	width=0.55 \textwidth]{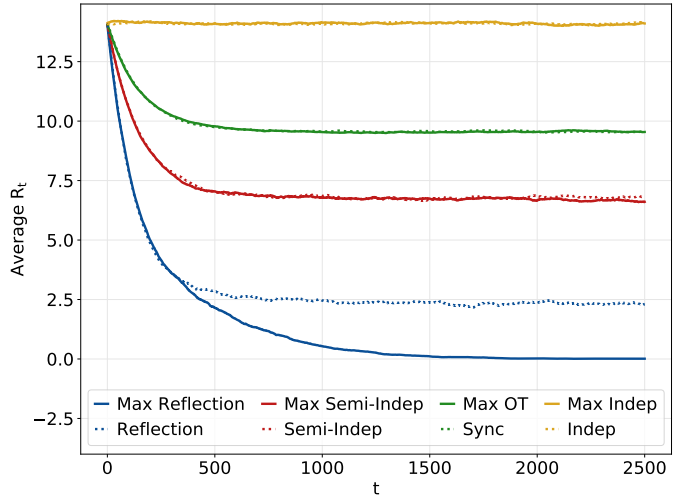}
	\caption{Average distance between chains over 1000 replications as a function of iteration $t$.
		Here $d=100$, and as in Figure~\ref{fig:ft:prop} these chains are initialized at independent draws
		from the target. At the meeting time $\tau$, chains switch to a sticky coupling in which $x'=y'$ and
		$U=V$. Maximal couplings and their non-maximal equivalents produce similar dynamics until $R_t$ is
		small enough for there to be a reasonable chance of meeting. This effect is visible for the
		reflection couplings but not the other options, which stop contracting before the chains are
		close enough for these effects to make a difference.
		\label{fig:mt:prop:rtrace}}
\end{figure}

The dynamics produced in this exercise provide a compelling explanation for the meeting time
behavior observed in Figure~\ref{fig:ft:prop}. In the absence of meeting, each coupling seems to
produce contraction down to a certain degree of separation between chains. For the maximal independent
coupling that appears to be almost exactly  $\E[\lVert Y_0 - X_0 \rVert]$, the distance obtained by
independent draws from the target distribution. The maximal optimal transport coupling and maximal
semi-independent coupling produce contraction to within a smaller radius. The explosive increase in
meeting times under these couplings suggests that these critical distances do not keep pace with the
$\bigO(1/\sqrt{d})$ rate noted above. By the same token, the maximal reflection coupling appears to
produce sufficient contraction to eventually meet with high probability. Among the four
maximal couplings considered here and their four non-maximal counterparts, only the maximal reflection coupling produces a high enough meeting probability to eventually diverge from its non-maximal counterpart.

We can also visualize these differences in drift directly, by creating pairs $(X_t,Y_t)$ with a
specific $R_t = r$, running a single step of the coupled MH kernel, and recording the resulting distance
$R_{t+1}$. We show the output of such a test in Figure~\ref{fig:mt:prop:contract}. Again we set
$d=100$ and use the $U=V\sim \Unif$ coupling at the acceptance step. We consider a range of $r$
values and initialize $\xy$ to have $e=(1,0,\dots)$, $m_1 = 1$, and $||m|| = \sqrt{d}$. In this case
we run 10,000 replications for each coupling and $r$ value. Consistent with the results above, we
find that the different proposal couplings display a range of contraction behavior as a function of
the distance between the chains, although all are contractive when the chains are far apart and
repulsive when the chains are close together. Except for the reflection coupling, the $R_t$ value
where each contraction line crosses the x-axis corresponds to the long-run average value of $R_t$,
as one would expect for a chain in close to a stable equilibrium around this point.

We conclude by noting that the meeting time, separation, and drift behavior illustrated in the plots
above agrees with our expectations in some cases more than others. For instance, it is not
surprising that using a maximal coupling with independent residuals in the proposal step produces
poor contraction behavior as a function of dimension. Although the proposal variance $\sigma^2_d$
shrinks in $d$, this coupling produces almost independent values of $X$ and $Y$ conditional on $x'
\neq y'$, whose separation can be expected to increase linearly in the number of independent
dimensions. Thus the flat line in Figure~\ref{fig:mt:prop:rtrace} agrees with intuition.

Each of the other three couplings has the property that $x\n' = y\n'$ when $x' \neq y'$, which
limits the potential for variance from these components as a function of dimension. However, the
relative performance of the maximal semi-independent, maximal optimal transport, and maximal reflection
couplings is almost the opposite of what one might expect. The optimal transport coupling seems to
produce the least contraction in spite of producing the smallest values of $\lVert x' - y' \rVert$
conditional on $x' \neq y'$. At the same time, the reflection coupling appears to produce the most
contraction despite maximizing the variance of $y'_1 - x'_1$. These differences seem likely to stem from the interaction of these couplings with the acceptance step.

\begin{figure}
	\centering
	\includegraphics[trim={.4cm .5cm 2.75cm 2.5cm}, clip=true,
	width=0.55\textwidth]{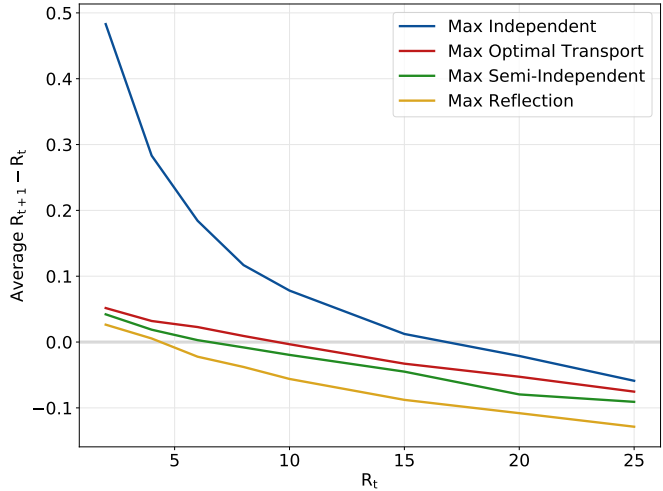}
	\caption{Average contraction as a function of the current distance between chains for four maximal couplings.
		For each point $R_t=r$ we construct a state pair such that $\lVert Y_t - X_t \rVert = r$. We then run one RWM iteration with the specified proposal coupling and the $U=V$ acceptance coupling, and record the resulting $R_{t+1} = \lVert Y_{t+1} - X_{t+1} \rVert$. We compute averages over 10,000 replications for each coupling and $r$ to obtain the curves shown here. The depicted drift behavior appears consistent with the meeting times and time series dynamics of Figures~\ref{fig:ft:prop} and \ref{fig:mt:prop:rtrace}.
		\label{fig:mt:prop:contract}
	}
\end{figure}

\subsection{Acceptance couplings}

Next we consider couplings of the accept/reject step. As noted in Section~\ref{sec:accrej}, we focus
on acceptance indicator couplings that accept both chains at exactly the MH rate for any pair of
proposal states. In light of Lemma~\ref{lem:unifcpl}, it is convenient to define acceptance
indicators $b_x = 1(U \leq a(x,x'))$ and $b_y = 1(V \leq a(y,y'))$ in terms of underlying uniform
random variables. We can realize three basic couplings by drawing $U \sim \Unif$ and then either
drawing $V \sim \Unif$ independently, setting $V = U$, or setting $V=1-U$. The $V=U$ coupling
maximizes the probability of $b_x = b_y$ while the $V=1-U$ coupling minimizes it. We also consider
the `optimal transport' acceptance indicator coupling described in Section~\ref{sec:accrej}. We
recall that this option almost surely coincides with either the $V=U$ or $V=1-U$ coupling at each
iteration, depending on which of these minimizes the expected distance $\E[ \lVert Y - X \rVert^2
\g x,y,x',y']$.

\begin{figure}
	\centering
	\subfloat[
	\label{fig:mt:accrej_refl:linear} Linear scale]{
		\includegraphics[trim={.8cm .5cm 2.75cm 2.5cm}, clip=true,
		width=0.45\textwidth]{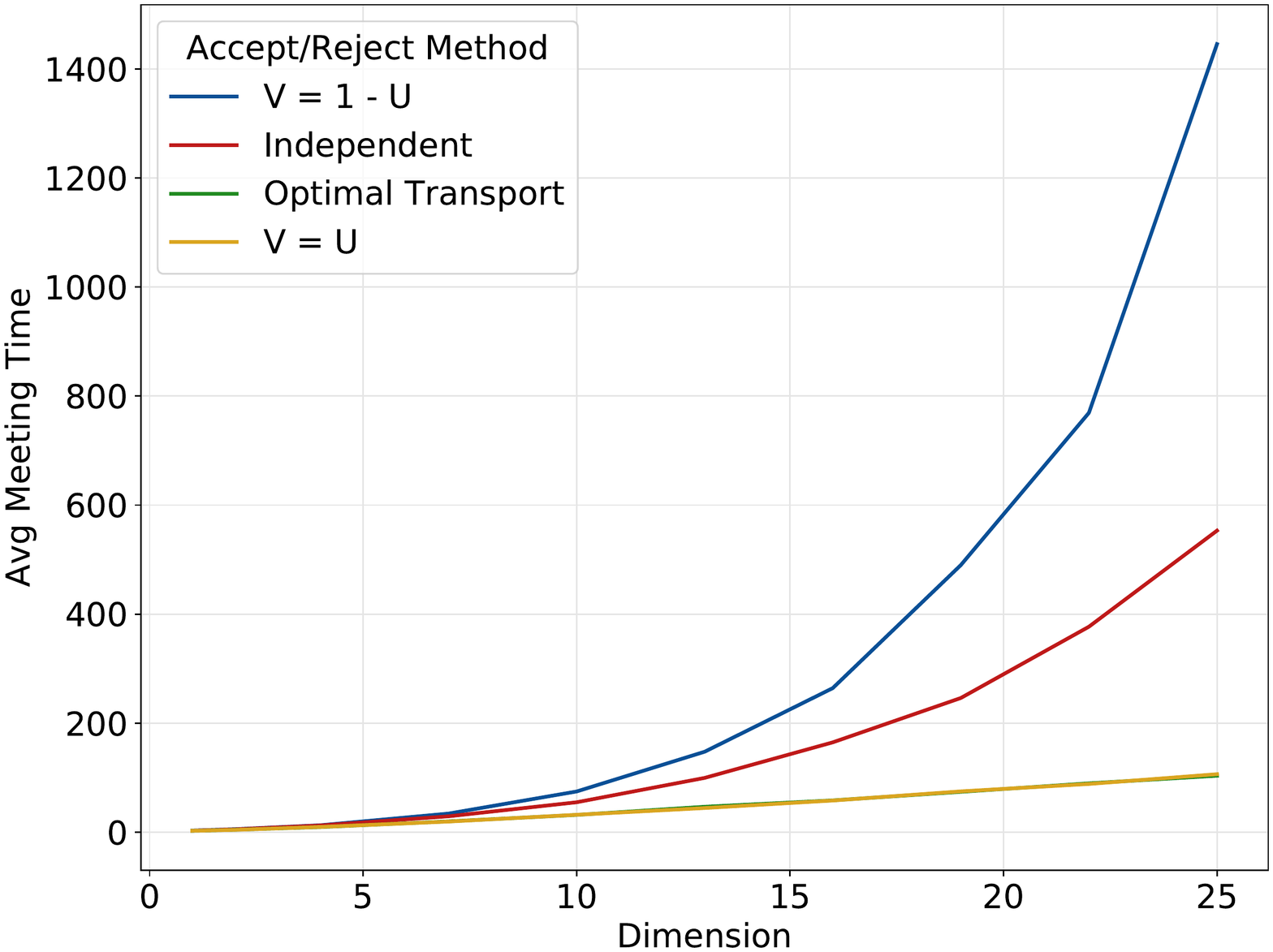} }
	\subfloat[
	\label{fig:mt:accrej_refl:loglog} Log scale]{
		\includegraphics[trim={.8cm .5cm 2.75cm 2.5cm}, clip=true,
		width=0.45\textwidth]{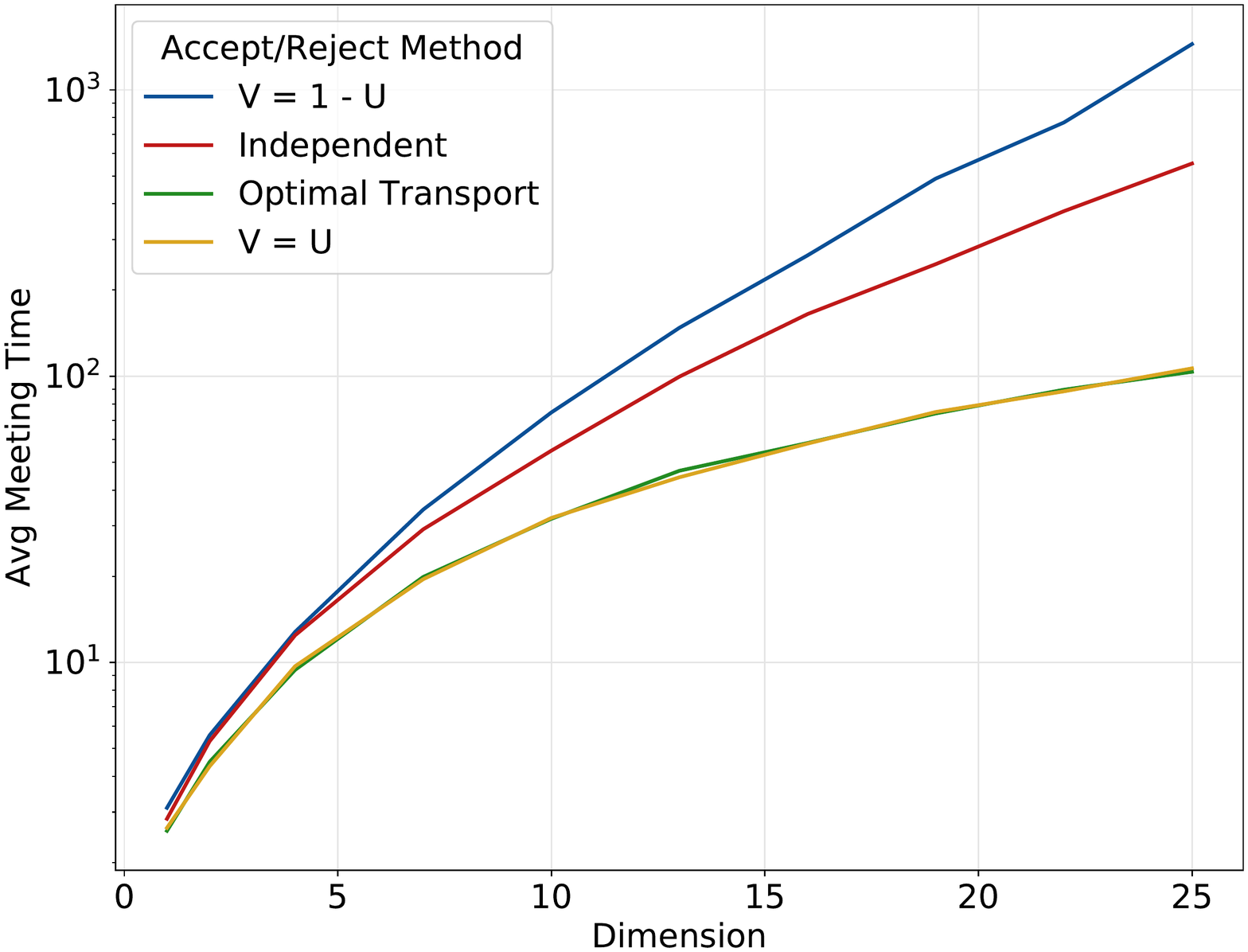} }
	\caption{\label{fig:mt:accrej_refl}~ Scale behavior of the average meeting time of coupled MH
		chains as a function of the dimension of the target distribution $\pi = \N(0, I_d)$, under
		various acceptance step couplings defined in Section~\ref{sec:accrej}. Here the chains are
		initialized at independent draws from the target, the proposal distributions are related by a
		maximal reflection coupling, and averages are taken over 1000 replications for each dimension and coupling. Left: the $V=1-U$ and independent $V,U$ methods produce meeting times which grow rapidly in dimension, in contrast to the $V=U$ and optimal transport couplings, which produce approximately linear behavior. Right: A log scale plot suggests that the $V=1-U$ and independent methods grow at less than an exponential rate, which would appear as a straight line on this plot.
	}
\end{figure}

We present simulation results on these four acceptance step couplings in
Figure~\ref{fig:mt:accrej_refl}. In each case we use the maximal reflection coupling of proposal
distributions and initialize using independent draws from the target. In
Figure~\ref{fig:mt:accrej_refl:linear}, we see that the $V=U$ coupling produces meeting times that
scale approximately linearly in dimension, while these increase more rapidly under the independent
and $V = 1-U$ couplings. The log-linear plot in Figure~\ref{fig:mt:accrej_refl:loglog} suggests that
these latter meeting times may still be less than exponential in dimension.

\begin{figure} \centering \subfloat[\label{fig:mt:accrej_max:choice} Optimal transport coupling
	choice]{ \includegraphics[trim={.8cm .5cm 2.75cm 2.5cm}, clip=true,
		width=0.45\textwidth]{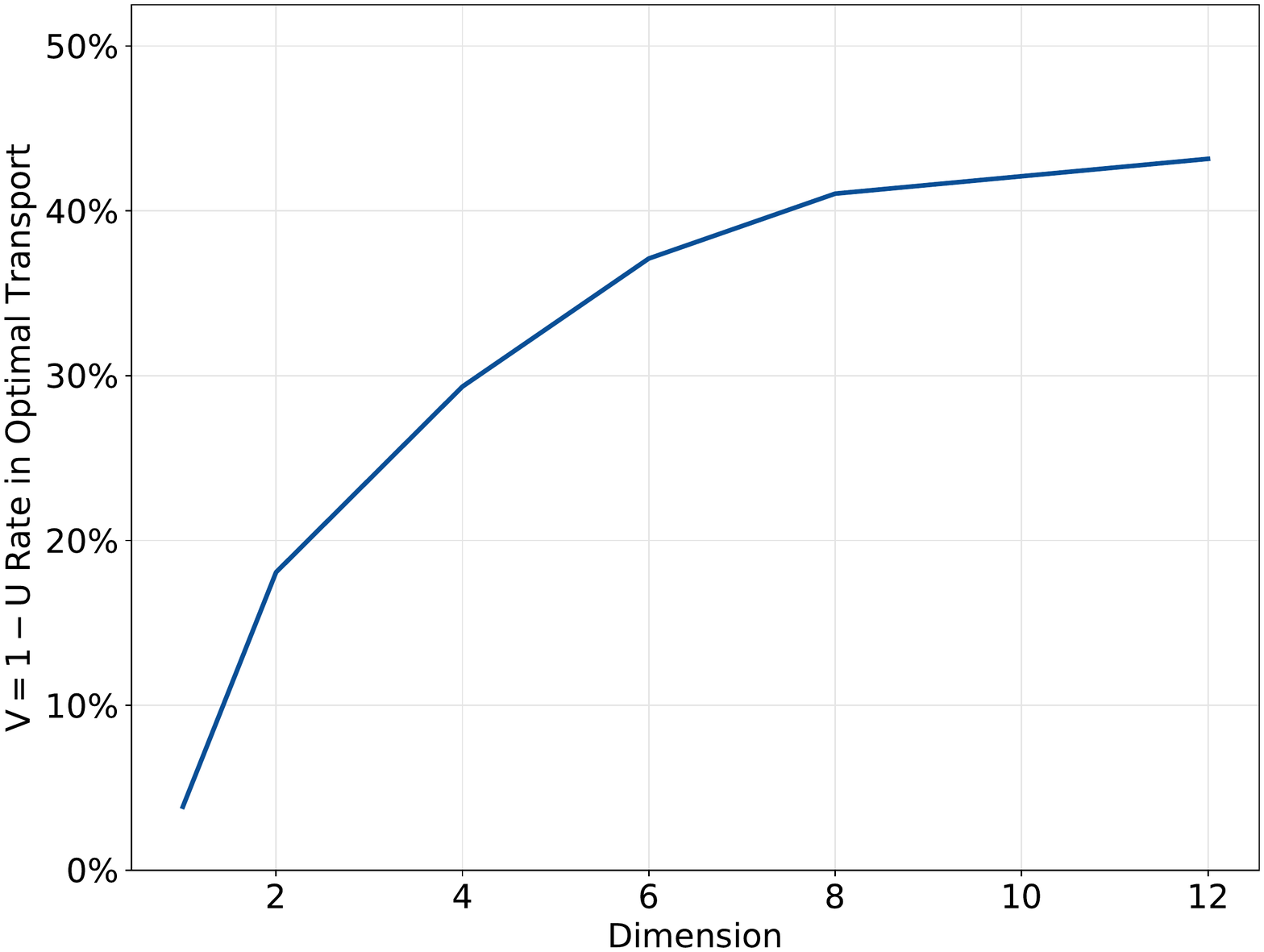} }
	\subfloat[\label{fig:mt:accrej_max:log} Meeting times (log scale)]{ \includegraphics[trim={.8cm .5cm
			2.75cm 2.5cm}, clip=true, width=0.45\textwidth]{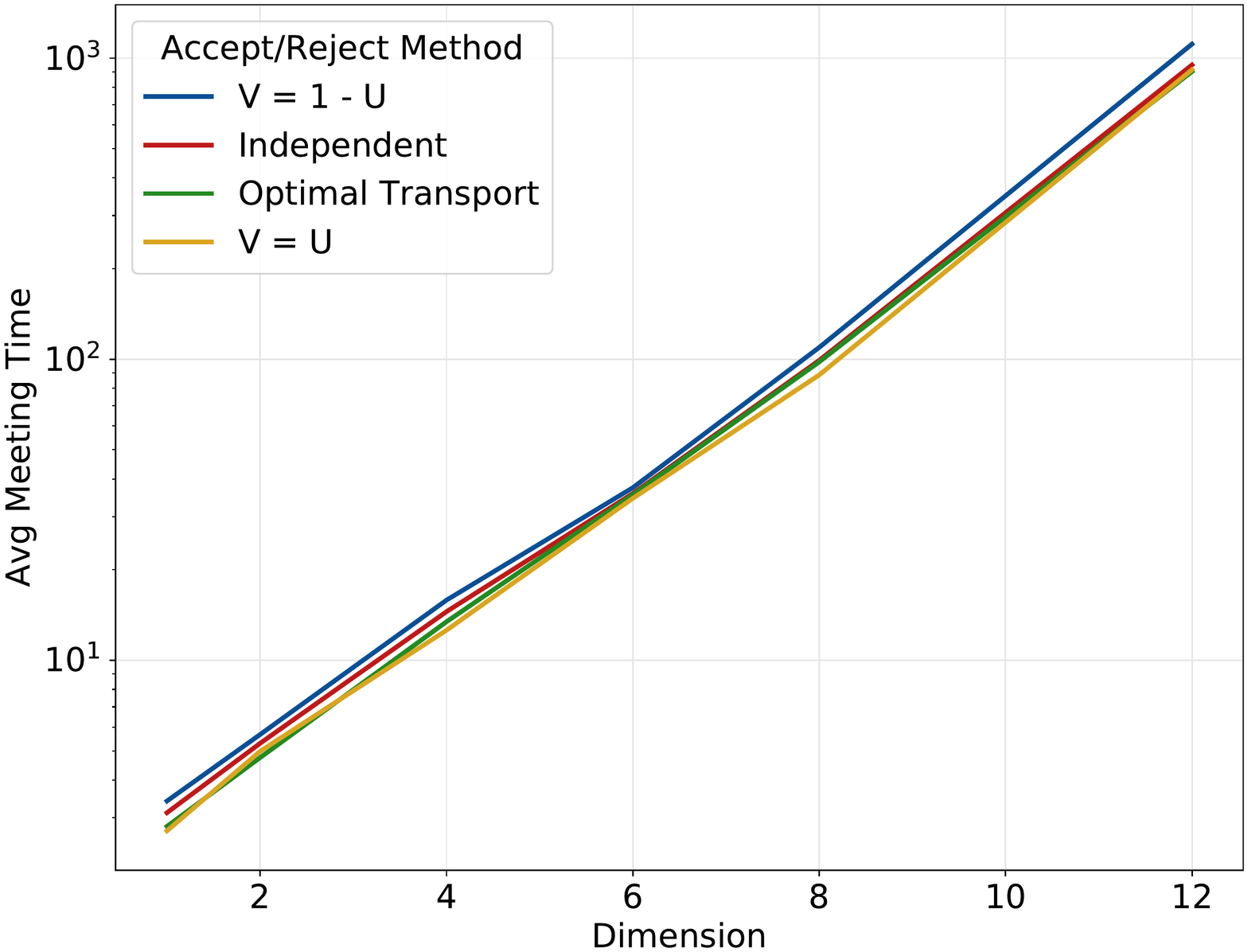} }
	\caption{\label{fig:mt:accrej_max}~Behavior of the optimal transport acceptance step coupling,
		as a function of the dimension of the target distribution
		$\pi = \N(0, I_d)$. Here the proposal distributions are related by a maximal coupling with \textit{independent}
		residuals, and as usual the chains are initialized by independent draws from the target distribution and replicated 1000 times per coupling and dimension.
		Left: the optimal transport coupling coincides with the $V=1-U$ coupling at a rate approaching 50\% as
		the dimension $d$ of the target increases. Right: all four acceptance step couplings deliver similar,
		exponentially growing meeting times under the maximal independent coupling of proposal
		distributions. }
\end{figure}

The optimal transport and $U=V$ couplings produce nearly identical results when applied to the
maximal reflection coupling of proposal distributions. A closer look at this scenario reveals that
the optimal transport coupling coincides with the $U=V$ coupling in all 1000 replications when $d>1$
and in approximately 96.6\% of replications in the $d=1$ case. We observe qualitatively identical
behavior when the proposals are maximally coupled with optimal transport and semi-independent
residuals.

Figure~\ref{fig:mt:accrej_max} shows that the acceptance step optimal transport coupling
displays more complex behavior when the proposal distributions follow a maximal coupling with
independent residuals. Here $V = 1 - U$ is optimal in a fraction of iterations going to  50\% as $d$
increases. Nevertheless, the resulting meeting times are almost indistinguishable from the meeting times
delivered by the $V=U$ and $V=1-U$  couplings. This suggests that under the maximal coupling with
independent residuals, the rapid growth in meeting times is due to the proposal coupling
more than any particular choice of acceptance indicator coupling.

\begin{figure}
	\centering
	\subfloat[\label{fig:mt:accrej:rtrace} Average separation by iteration $t$]{
		\includegraphics[trim={.8cm .5cm 2.75cm 2.5cm}, clip=true, width=0.45\textwidth]{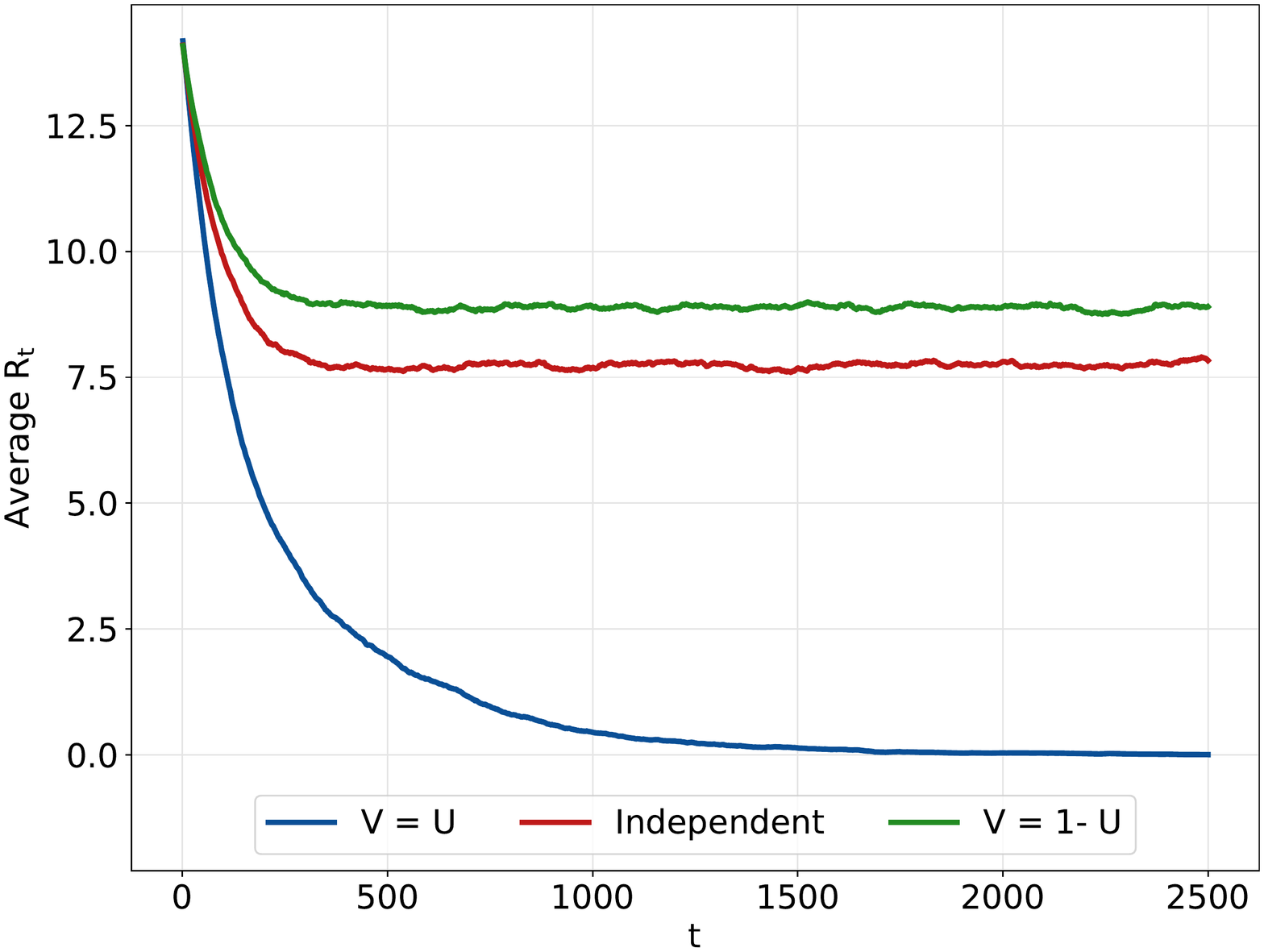} }
	\subfloat[\label{fig:mt:accrej:drift} Contraction vs. distance between chains ]{
		\includegraphics[trim={.5cm .5cm 2.75cm 2.5cm}, clip=true,
		width=0.45\textwidth]{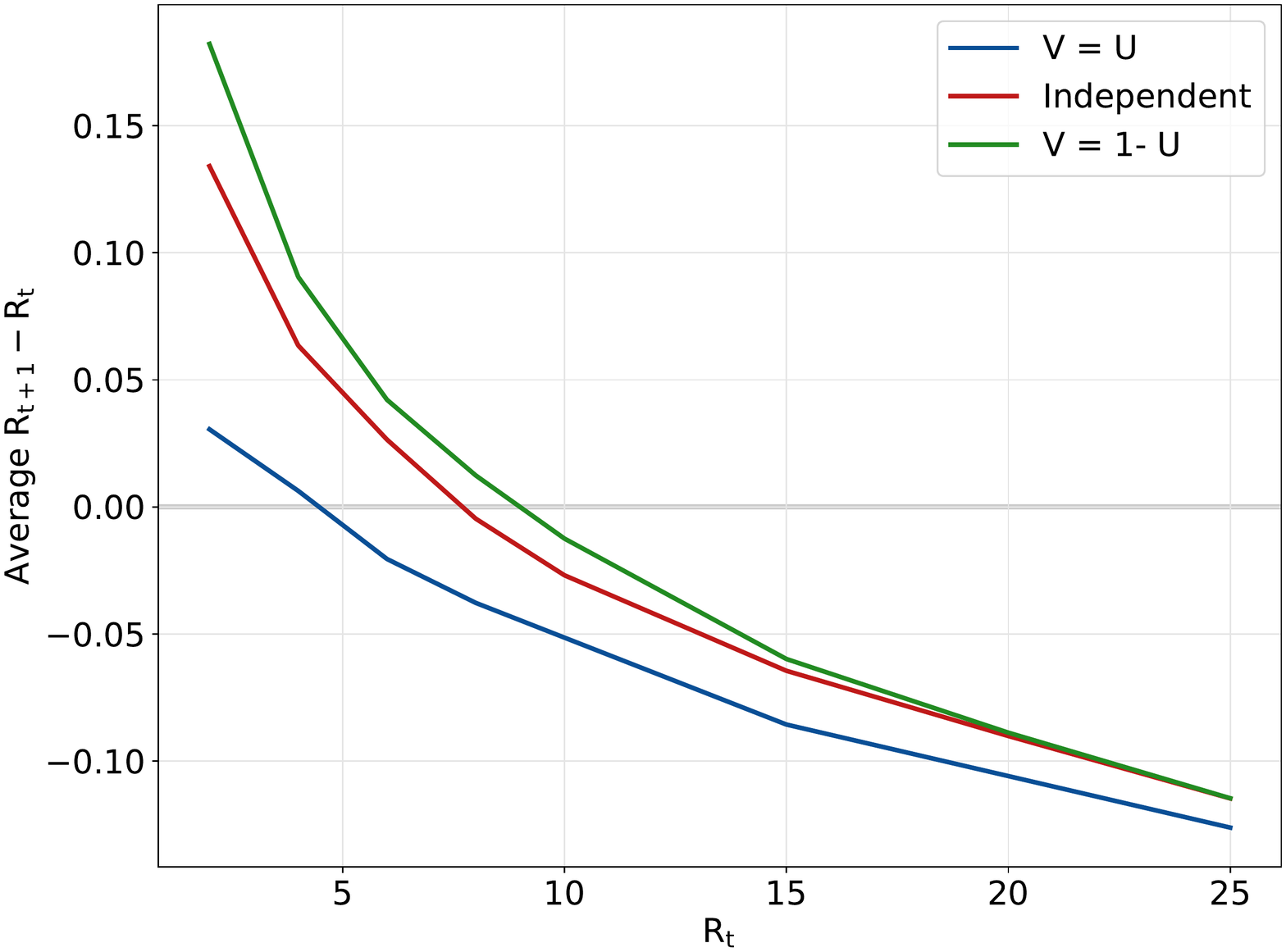} }
	\caption{\label{fig:mt:accrej}~ Time series and drift function properties of the distance between
		chains under the three simple acceptance step couplings. Left: we repeat the experiment shown in
		Figure~\ref{fig:mt:prop:rtrace} for the acceptance step couplings. Only the $V=U$ option allows the chains to get close enough to produce a significant probability of meeting. Right: we repeat the experiment shown in
		Figure~\ref{fig:mt:prop:contract} for these couplings. The $V=U$ coupling displays more contraction and a smaller $x$-intercept than the other options. The behavior of this intercept as a function of $d$ has a significant effect on meeting times.}
\end{figure}

As in the case of proposal couplings, we can also understand the meeting times associated with
different acceptance indicator couplings in terms of the contraction between chains. In
Figure~\ref{fig:mt:accrej:rtrace}, we present the average distance between chains under the three
simple acceptance indicator couplings. We run all of these using the maximal reflection coupling of
proposal distributions. As in the proposal coupling case we set $d=100$, we initializing each chain
with an independent draw from the target, and we use a sticky coupling to ensure $X_t=Y_t$ for $t
\geq \tau$. The $U=V$ coupling is able to produce sufficient contraction for meeting to take place,
while this seems out of reach for the independent and $V=1-U$ couplings.

We also consider the effect of different acceptance couplings on the drift $R_{t+1}-R_t$ as a
function of the current distance between chains $R_t$. We use the maximal reflection coupling at the
proposal step, and we run 10,000 replications for each value of $R_t$ and coupling option. As in the
case of proposal couplings, the time series behavior of $R_t$ shown in
Figure~\ref{fig:mt:accrej:rtrace} is consistent with relationship between $R_t$ and
$\E[R_{t+1}-R_{t} \g X_t, Y_t]$ observed in Figure~\ref{fig:mt:accrej:drift}.

Both of these tests support the impression that the choice of the acceptance coupling has a
significant effect on the contraction properties of the resulting chains. At a high level, it
appears that the right combination of proposal and acceptance strategies can lead to powerful
contraction between chains down to a point where meeting is reasonably probable under a maximal
coupling. The combination of a maximal reflection proposal coupling and a maximal acceptance coupling has this property while most other combinations do not, leading to a rapid growth in meeting times as a function of dimension.

\section{Discussion}

In the sections above we have identified a range of options for use in the design of RWM transition
kernel couplings. Our analysis and simulations suggest a few principles for the choice of these
elements, which we summarize as follows.

First, the coupling inequality imposes a significant constraint on the ability of any $\bar Q \in
\Gamma(Q,Q)$ to propose meetings. This suggests using a maximal or nearly maximal coupling to obtain
meetings at the highest rate possible. A hybrid approach may also be practical in some cases.
When a meeting is not proposed, it seems advantageous to minimize the degrees of freedom
in the displacement $y'-x'$ between proposals. These degrees of freedom accumulate in higher dimensions
and eventually create a barrier to contraction between chains. This may explain the poor performance of the
maximal independent coupling relative to the maximal semi-independent coupling.

Since the probability of a meeting is typically small until the chains are close together, it is
important to construct a transition kernel coupling that yields strong and persistent contraction between chains.
Surprisingly, the reflection couplings seem to do the best job of this among the proposal options
considered above. These couplings do not have good contraction properties on their own, but they
seem to set up a favorable interaction with the Metropolis step, especially with the $U=V$ coupling.
The precise nature of this interaction is an important open question. For now, it
appears safe to recommend the reflection coupling for inducing contraction
between chains.

The success of the reflection coupling raises two additional questions. First, we may consider the
extent to which this behavior depends on the log-concavity of the target distribution.
It seems reasonable to think that this coupling may not work as well with irregular
targets. With log-concave targets, like $\pi_d = \N(0, I_d)$, we can also ask how close the MH transition
kernels based on a maximal coupling with reflection residuals at the proposal step comes to a maximal coupling with optimal transport residuals of the transition kernels themselves. This question seems amenable to either theoretical
and numerical methods.

On the acceptance indicator side, the $U=V$ coupling has a strong a priori appeal. This coupling
gives the highest chance of turning a proposed meeting into an actual meeting. It also minimizes the
probability of accepting one proposal and rejecting the other, which often leads to a jump in the
distance between chains. While the $U=V$ coupling dominates the other options in this study, we recall
that we have focused our attention on the subset of acceptance indicator couplings in which the
conditional acceptance rates $a_x(\xy,\xyp)$ and $a_y(\xy,\xyp)$ agree with the MH rates $a(x,x')$
and $a(y,y')$. The analysis of more general acceptance couplings deserves further attention.

We emphasized the simple case of a multivariate normal target distribution in the simulations above.
It would be interesting to know the extent to which our conclusions generalize to more challenging
examples such as targets with heavy tails, multi-modality, difficult geometries, and examples in
large discrete state spaces. One might also extend the coupling strategies described above to other
common MH algorithms such HMC \citep{Duane:1987, Neal1993, neal2011mcmc}, the Metropolis-adjusted
Langevin algorithm \citep{roberts1996exponential}, and particle MCMC
\citep{andrieu:doucet:holenstein:2010}. We expect that couplings for these extensions would involve
some of the same principles as above, but with more moving parts and fewer symmetries to exploit.

Perhaps the most important open questions in the area of coupling design concern the development of
theoretical tools to relate proposal and acceptance options to meeting times. Such tools would enable
a systematic understanding of the interaction between proposal and acceptance steps. This would also
support work on how to pair these to produce as much possible contraction as possible between chains. One
approach to this might exploit the drift and minorization approach of \citet{Rosenthal1995,
	rosenthal2002quantitative}, especially the pseudo-small set concept of \citet{roberts2001small}. The
analyses and simulations above mark a step forward in our understanding of the options for coupling
MH transition kernels. They suggest that some options might be better than others and hint at
why.

\subsubsection*{Acknowledgements}
The author thanks Pierre E. Jacob, Yves Atchad\'e, and Niloy Biswas for their helpful comments. He also gratefully acknowledge support by the National Science Foundation through grant DMS-1844695.


\bibliography{dissertation_refs}{}




\end{document}